\documentclass[a4paper,11pt]{amsart}
\usepackage{geometry}
\usepackage{amssymb, colordvi}
\usepackage{multirow}
\usepackage{mathrsfs}
\usepackage[pdftex]{graphicx,color}
\usepackage{tikz}
\usepackage{pgf}
\usepackage{pgffor}
\usetikzlibrary{arrows,positioning,matrix,decorations.pathmorphing,calc,fadings,decorations.pathreplacing} 
\usepackage[hidelinks]{hyperref}
\usepackage{todonotes} 
\usepackage{enumerate}
\usepackage{comment}
\usepackage{enumitem}
\usepackage{makecell}

\numberwithin{equation}{section}
\newtheorem{thm}{Theorem}

\numberwithin{theorem}{section}
\numberwithin{thm}{section}
\newtheorem{prop}[thm]{Proposition}
\newtheorem{proposition}[thm]{Proposition}
\newtheorem{lemma}[thm]{Lemma}
\newtheorem{corollary}[thm]{Corollary}

\newtheorem{remark}[thm]{Remark}
\newtheorem{definition}[thm]{Definition}
\newtheorem{algorithm}[thm]{Algorithm}

\newcounter{FNC}[page]
\def\fauxfootnote#1{{\addtocounter{FNC}{2}$^\fnsymbol{FNC}$%
     \let\thefootnote\relax\footnotetext{$^\fnsymbol{FNC}$#1}}}

\newcommand{\vol}{{\rm vol}}

\newcommand{\R}{\mathbb{R}}
\newcommand{\C}{\mathbb{C}}

\newcommand{\Z}{\mathbb{Z}}

\newcommand{\A}{{\mathcal A}}

\DeclareMathOperator{\minor}{minor}

\DeclareMathOperator{\sign}{sign}

\DeclareMathOperator{\ch}{ch}

\newcommand{\arrowschem}[2]{\raisebox{-2ex}%
	{$\stackrel{\stackrel{\displaystyle#1}{\longrightarrow}}%
	{\stackrel{\longleftarrow}{#2}}$}}

\begin{document}

\title[Regions with many positive steady states]
{Parameter regions that give rise to $2 [\frac{n}{2}]+1$ positive steady states in the $n$-site phosphorylation system}

\author{Magal\'i Giaroli}
\address{Dto.\ de Matem\'atica, FCEN, Universidad de Buenos Aires, and IMAS (UBA-CONICET), Ciudad Universitaria, Pab.\ I, 
C1428EGA Buenos Aires, Argentina}
\email{mgiaroli@dm.uba.ar}

\author{Rick Rischter}
\address{Instituto de Matem\'atica e Computa\c{c}\~ao, IMC, Universidade Federal de Itajub\'a (UNIFEI)\\ 
Av. BPS 1303, Bairro Pinheirinho\\ 
37500-903, Itajub\'a, Minas Gerais\\ 
Brazil}
\email{rischter@unifei.edu.br}
\urladdr{http://w3.impa.br/~rischter/}

\author{Mercedes P\'erez Mill\'an}
\address{Dto.\ de Matem\'atica, FCEN, Universidad de Buenos Aires, and IMAS (UBA-CONICET), Ciudad Universitaria, Pab.\ I, 
C1428EGA Buenos Aires, Argentina}
\email{mpmillan@dm.uba.ar}
\urladdr{http://cms.dm.uba.ar/Members/mpmillan}

\author{Alicia Dickenstein}
\address{Dto.\ de Matem\'atica, FCEN, Universidad de Buenos Aires, and IMAS (UBA-CONICET), Ciudad Universitaria, Pab.\ I, 
C1428EGA Buenos Aires, Argentina}
\email{alidick@dm.uba.ar}
\urladdr{http://mate.dm.uba.ar/~alidick}
\date{}


\begin{abstract} 
The distributive sequential $n$-site phosphorylation/dephosphorylation system 
is an important building block in networks of chemical reactions arising in molecular biology, 
which has been intensively studied. 
 In the nice paper of Wang and Sontag (2008) it is shown that for certain choices
of the reaction rate constants and total conservation constants, the system can have $2 [\frac{n}{2}]+1$
positive steady states (that is,
$n+1$ positive steady states for $n$ even and $n$ positive steady states for $n$ odd). In
this paper we give open parameter regions in the space of  reaction rate constants and  total
conservation constants that ensure these number of positive steady states, while assuming in the modeling 
that roughly only $\frac 1 4$ of the intermediates occur in the reaction mechanism. This result is based on the general framework
developed by Bihan, Dickenstein, and Giaroli (2018), which can be applied to
other networks.  We also describe how to implement these tools to search for multistationarity regions in a computer
algebra system and present some computer aided results.
  \end{abstract}
\maketitle

\section{Introduction}
Multisite phosphorylation cycles are common in cell regulation systems~\cite{lim14}. 
The important distributive sequential $n$-site phosphorylation network is
given by the following reaction scheme:
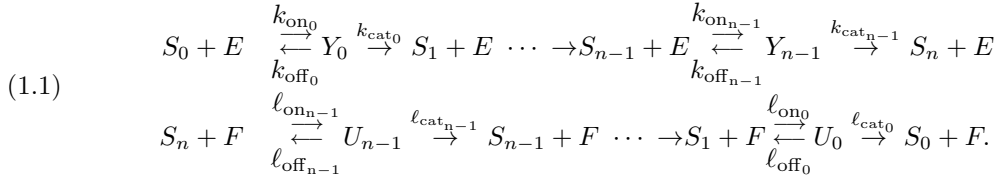
\begin{small}
  \begin{equation}  \label{eq:nsite}
 \begin{array}{rl} 
 S_0+E &
 \arrowschem{k_{\rm{on}_0}}{k_{\rm{off}_0}} Y_0
 \stackrel{k_{\rm{cat}_0}}{\rightarrow} S_1+E\ \cdots\ {\rightarrow}S_{n-1}+E 
 \arrowschem{k_{\rm{on}_{n-1}}}{k_{\rm{off}_{n-1}}}Y_{n-1}
 \stackrel{k_{\rm{cat}_{n-1}}}{\rightarrow} S_n+ E \\
 S_n+F &
 \arrowschem{\ell_{\rm{on}_{n-1}}}{\ell_{\rm{off}_{n-1}}} U_{n-1}
 \stackrel{\ell_{\rm{cat}_{n-1}}}{\rightarrow} S_{n-1}+F \ \cdots\ {\rightarrow} S_{1}+F
 \arrowschem{\ell_{\rm{on}_0}}{\ell_{\rm{off}_0}} U_0
 \stackrel{\ell_{\rm{cat}_0}}{\rightarrow} S_0+ F.
 \end{array}
 \end{equation}
 \end{small}%
Arrows correspond to chemical reactions. The labels in the arrows are positive numbers, known as {\em reaction rate constants}.
This reaction scheme represents one substrate  $S_0$ that can sequentially acquire up to $n$ phosphate groups 
(giving rise to the phosphoforms $S_1, \dots, S_n$) 
via the action of the kinase $E$, and which can be sequentially released via the action of the phosphatase $F$, in both cases via an intermediate species (denoted by $Y_0$, $Y_1$, $\dots$, $Y_{n-1}$, $U_0$, $U_1, \dots, U_{n-1}$)
formed by the union of the substrate and the enzyme . 

  The kinetics of this network is deduced by applying the law of mass action to the labeled digraph of reactions~\eqref{eq:nsite}, 
which yields an autonomous polynomial system of ordinary differential equations depending on the positive reaction rate
constants. This system describes the evolution in time of the concentrations of the different chemical species. We denote
 with lower case letters the concentrations of the $3n+3$ species. The derivatives of the concentrations
of the substrates satisfy (see e.g.~\cite{alicia} for the general definition):
  \begin{small}
 \begin{align*}\label{eq:mak}
 \frac{ds_0}{dt} &=  {-k_{\rm{on}_0}}s_0e + {k_{\rm{off}_0}}y_0 + {\ell_{\rm{cat}_0}} u_0, \nonumber \\
 \frac{ds_i}{dt} &=  {k_{\rm{cat}_{i-1}}}y_{i-1} - {k_{\rm{on}_{i}}}s_ie + {k_{\rm{off}_{i}}}y_{i} + {\ell_{\rm{cat}_i}}u_i -
  {\ell_{\rm{on}_{i-1}}}s_if + {\ell_{\rm{off}_{i-1}}}u_{i-1},\quad i=1, \dots, n-1, \nonumber \\
 \frac{ds_n}{dt}& = {k_{\rm{cat}_{n-1}}}y_{n-1}  -{\ell_{\rm{on}_{n-1}}}s_n f + {\ell_{\rm{off}_{n-1}}} u_{n-1},
 \end{align*}
  \end{small}%
the derivatives of the concentrations of the intermediate species satisfy:
\begin{small}
 \begin{align*}
 \frac{dy_i}{dt} &=  {k_{\rm{on}_i}}s_i e -{(k_{\rm{off}_i}+k_{\rm{cat}_i})}y_i, \quad i=0,\dots,n-1,  \\
  \frac{du_i}{dt} &=  {\ell_{\rm{on}_i}}s_{i+1} f -{(\ell_{\rm{off}_i}+\ell_{\rm{cat}_i})}u_i, \quad i=0,\dots,n-1,  
  \end{align*}
  \end{small}%
and $ \frac{de}{dt}= -\frac{dy_0}{dt} - \dots - \frac{dy_{n-1}}{dt}$,\,  $\frac{df}{dt}= -\frac{du_0}{dt} -\dots - \frac{du_n}{dt}$,
  which give two linear conservation relations. Indeed, 
 there are three linearly independent conservation laws:
 \begin{small}
 \begin{equation}\label{eq:consfosfo}
 \sum_{i=0}^n s_i + \sum_{i=0}^{n-1} y_i + \sum_{i=0}^{n-1} u_i=  S_{tot}, \quad
 e + \sum_{i=0}^{n-1} y_i=  E_{tot}, \quad
 f + \sum_{i=0}^{n-1} u_i=  F_{tot},
\end{equation}
\end{small}%
where the {\em linear conservation constants}
$S_{tot}, E_{tot}, F_{tot}$ are positive for any trajectory of the system intersecting the positive orthant.

A steady state of the system is a common zero of the polynomials in the right-hand side of the
differential equations.
There are $6n+3$ parameters: the reaction rate constants and the linear conservation constants  (which correspond to points
in the positive orthant $\R_{>0}^{6n+3}$).
It is known that for any $n \ge 2$, the sequential $n$-site  system shows {\em multistationarity}. This means that there exists  a choice of parameters
 for which there are at least two positive steady states satisfying moreover the linear conservation relations~\eqref{eq:consfosfo} for the same $S_{tot},E_{tot}$ and $F_{tot}$. In this case, it is said that the steady states are {\em stoichiometrically compatible}, or that they lie in the
 same {\em stoichiometric compatibility class}.
 This is an important notion, because the occurrence of multistationarity allows  for different cell responses with the same
 linear conservation constants (that is, the same total amounts of substrate and enzymes). 
 
 The possible number of positive steady states of the $n$-site phosphorylation system  (for fixed linear conservation constants)
  has been studied in several articles. 
 For $n=2$, it is a well known fact that the number of nondegenerate positive steady states is one or three~\cite{Markevich04,sontag}. 
 The existence of bistability is proved in~\cite{Hell15}.  
 In \cite{CFMW17} and \cite{CM14}, the authors give conditions on the reaction rate constants to 
 guarantee the existence of three positive steady states based on
 tools from degree theory, but this approach does not describe the linear conservation constants for which there is multistationarity. 
 For an arbitrary number $n$ of phosphorylation sites, it was shown in \cite{sontag} that the system has at most $2n-1$ positive steady states. 
In the same article, the authors showed that there exist reaction rate constants and linear conservation constants such that the network has
  $n$ (resp. $n+1$) positive steady states for $n$ odd (resp. even);
  that is, there are $2[\frac{n}{2}]+1$ positive steady states for any value of $n$, where $[.]$ denotes integer part.
  
 In~\cite{2n-1} (see also~\cite{kfc}) the authors showed parameter values such that for $n=3$ the system has five 
 positive steady states, and for $n=4$ the system has seven steady states, obtaining the upper bound given in~\cite{sontag}. 
 In the recent article~\cite{CAT18} the authors show that if the $n$-site phosphorylation system  is multistationary for a choice
 of rate constants and linear conservation constants $(S_{tot}, E_{tot}, F_{tot})$ then for any positive 
 real number $c$ there are rate constants for which the system is multistationary when the linear conservation 
 constants are scaled by $c$. Concerning the stability, in \cite{thomson09} it is shown evidence that the system 
 can have $2[\frac{n}{2}]+1$ positive steady states with $[\frac{n}{2}]+1$ of them being stable. Recently, 
 a proof of this unlimited multistability was presented in \cite{FRW19}, where the authors find a choice of 
 parameters that gives the result for a smaller system, and then extend this result using techniques from singular perturbation theory.
 
 In the previous paper~\cite{BDG1}, parameter regions on {\em all the parameters} are given for the occurrence of multistationarity for the $n$-site
 sequential phosphorylation system, but no more than three positive steady states are ensured. These conditions are 
 based on a general framework  to obtain multistationary regions jointly in the reaction rate constants and the 
 linear conservation constants.  Our approach in this article uses the systematic technique in~\cite{BDG1}, 
 which we briefly summarize in Section~\ref{sec:background}.

The removal of intermediates was studied in~\cite{fw13}.  More specifically, the emergence of 
multistationarity of the $n$-site phosphorylation system with less intermediates was studied in~\cite{SFeliu}. 
It is known that the $n$-site phosphorylation network without any intermediates complexes has only one steady state for any choice of parameters.
In~\cite{SFeliu}, the authors showed which are the minimal sets of intermediates that give rise to a multistationarity system, 
but they do not give information about how many positive steady states can occur, and also, they do not describe
 the parameter regions for which 
these subnetworks are multistationary.

In this paper, we work with subnetworks of the sequential $n$-site phosphorylation system that only have intermediates 
in the $E$ component (that is, in the connected component of the network where the kinase $E$ reacts), see Definition \ref{def:G_J}. 
In case of the full mechanism on the $E$ component or if we only assume that
there are intermediate species that are formed between the phosphorylated substrates with
{\em parity equal to $n$} (that is, roughly only $\frac{1}{4}$ of the intermediates of the $n$ sequential
phosphorylation cycle),
we obtain precise conditions on all the parameters to ensure as many positive steady states as possible. 
Indeed, we show in Proposition~\ref{thm:upperboundGJ}
that the maximum number of complex solutions to the steady state equations
intersected with the linear conservation relations is always $n+1$, the maximum number of real roots is also $n+1$, 
that could be all positive when $n$ is even, while 
only $n$ of them can be positive when $n$ is odd, so the maximum number of positive steady states
equals $2[\frac n 2]+1$ for any $n$.
Exact conditions on
the parameters so that the associated phosphorylation/dephosphorylation
system admits  these number of positive steady states, are given in Theorem~\ref{th:allintermediatesEside}
and Corollary~\ref{cor:1/4}. The latter follows from Theorem~\ref{th:J}. In order to state these results, we need to
introduce some notations.

\begin{definition}\label{def:G_J} 
For any natural number $n$, we write $I_n=\{0,\dots,n-1\}$. Given $n\geq 1$, and a subset $J\subset I_n$, we denote by $G_J$ 
the network whose only intermediate complexes are $Y_j$ for $j \in J,$ and none of the $U_i$. It is given by the following reactions:  
 {\small
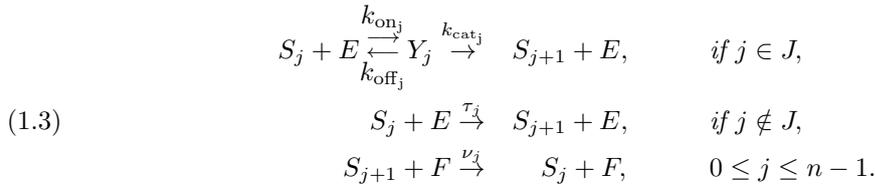
 \begin{eqnarray}   
 S_j+E 
 \arrowschem{k_{\rm{on}_j}}{k_{\rm{off}_j}} Y_j 
 \stackrel{k_{\rm{cat}_j}}{\rightarrow} &S_{j+1}+E,\quad \quad & \text{if } j\in J,\nonumber \\
 S_j+E   \stackrel{\tau_j}{\rightarrow} &S_{j+1}+E,\quad \quad &\text{if } j\notin J, \label{eq:net_phospho_n_lessintermediatesEside}\\
 S_{j+1}+F 
 \stackrel{\nu_{j}}{\rightarrow} & S_{j} +F, \quad & 0\leq j \leq n-1. \nonumber
\end{eqnarray}}%
where the labels of the arrows are positive numbers.
We will also denote by $G_J$ the associated differential system with mass-action kinetics.
\end{definition}

 For all these systems $G_J$, there are always three linearly independent conservation laws for any value of $n$:
 \begin{small}
 \begin{equation}\label{eq:consfosfoJ}
 \sum_{i=0}^n s_i + \sum_{j\in J} y_j =  S_{tot}, \quad
 e + \sum_{j\in J} y_j =  E_{tot}, \quad
 f =  F_{tot},
\end{equation}
\end{small}%
where the total conservation constants $S_{tot}, E_{tot}, F_{tot}$ are positive for
 any trajectory  of the differential system which intersects the positive orthant. 
 Note that the concentration of the phosphatase $F$ is constant, equal to $F_{tot}$.

To get lower bounds on the number of positive steady states with fixed positive linear conservation constants,
we first consider the network $G_{I_n}$, that is, when all the intermediates in the $E$ component appear.
It has the following associated digraph:
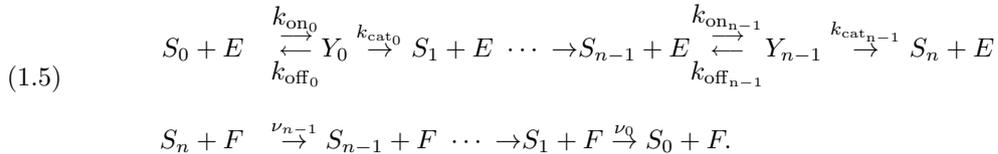
\begin{small}
  \begin{equation}
 \begin{array}{rl} 
 S_0+E &
 \arrowschem{k_{\rm{on}_0}}{k_{\rm{off}_0}} Y_0
 \stackrel{k_{\rm{cat}_0}}{\rightarrow} S_1+E\ \cdots\ {\rightarrow}S_{n-1}+E 
 \arrowschem{k_{\rm{on}_{n-1}}}{k_{\rm{off}_{n-1}}}Y_{n-1}
 \stackrel{k_{\rm{cat}_{n-1}}}{\rightarrow} S_n+ E \\
 \label{eq:net_phospho_n_complete_intermediatesEside} \\ 
 S_n+F &
 \stackrel{\nu_{n-1}}{\rightarrow} S_{n-1}+F \ \cdots\ {\rightarrow} S_{1}+F
  \stackrel{\nu_{0}}{\rightarrow} S_0+ F.
 \end{array}
 \end{equation}
 \end{small}

\medskip

\begin{thm}\label{th:allintermediatesEside}
 Let $n\geq 1$. With the previous notations, consider the network $G_{I_n}$ in \eqref{eq:net_phospho_n_complete_intermediatesEside},  
 and suppose that 
 the reaction rate constants $k_{\rm{cat}_i}$  and $\nu_i$, $i=0,\dots,n-1$, satisfy the inequality
\begin{small}
\[\max_{i \, \text{even}}\left\{\frac{k_{\rm{cat}_i}}{\nu_i}\right\}<\min_{i \, \text{odd}}\left\{\frac{k_{\rm{cat}_i}}{\nu_i}\right\}.\]  
\end{small}
For any positive values $S_{tot}$, $E_{tot}$ and $F_{tot}$ of the linear conservation constants with 
\begin{small}
\[ S_{tot}>E_{tot},\]
\end{small}%
verifying the inequalities:
\begin{small} \begin{equation}\label{ineq:allintermediates}
\max_{i \, \text{even}}\left\{\frac{k_{\rm{cat}_i}}{\nu_i}\right\}< 
\left(\frac{S_{tot}}{E_{tot}}-1\right)F_{tot}<\min_{i \, \text{odd}}
\left\{\frac{k_{\rm{cat}_i}}{\nu_i}\right\},
\end{equation}\end{small}%
there exist positive constants $B_1,\dots,B_n$ such that for any choice of positive constants $\lambda_0,\dots,\lambda_{n-1}$ satisfying
\begin{small}
\begin{equation}\label{ineq:gammaallintermediates2}
\frac{\lambda_j}{\lambda_{j-1}}<B_j \ \text{ for } j=1,\dots,n-1,\quad \frac{1}{\lambda_{n-1}}<B_n, 
\end{equation}
\end{small}%
rescaling of the given parameters $k_{\rm{on}_j}$ by $\lambda_{j}\, k_{\rm{on}_j}$, for each $j=0,\dots,n-1$, gives rise to a system with 
exactly $2[\frac{n}{2}]+1$ nondegenerate positive steady states.
\end{thm}


\begin{remark}\label{remark:rescaling} \rm We will also show in the proof of Theorem~\ref{th:allintermediatesEside}, 
 that for any reaction rate constants and linear conservation constants satisfying \eqref{ineq:allintermediates}, there exist $t_0>0$ 
 such that for any value of $t\in(0,t_0)$, the system $G_{I_n}$ has exactly $2[\frac{n}{2}]+1$ nondegenerate 
 positive steady states after modifying the constants $k_{\rm{on}_j}$ by $t^{j-n}k_{\rm{on}_j}$ for each $j=0,\dots,n-1$.
\end{remark}

We now consider  subnetworks $G_J$, with $J\subset J_n$ where
\begin{equation}\label{Jn}
J_n:=\{i\in I_n\,:\, (-1)^{i+n}=1\}, \text{ for } n\geq 1,
\end{equation}
that is, subsets $J$ with indexes that have the {\em same parity as $n$}.

\begin{thm}\label{th:J} Let $n\geq 1$, and consider a subset $J\subset J_n$. Let  $G_J$ be its associated system 
as in~\eqref{eq:net_phospho_n_lessintermediatesEside}.
Assume moreover that  
\begin{small}
\[ S_{tot}>E_{tot}.\] 
\end{small}%
A multistationarity region in the space of all parameters for which the system $G_J$ admits at least $1+2|J|$ positive steady states can be described
as follows.
Given any positive value of $F_{tot}$, choose any positive real numbers $k_{\rm{cat}_j}, \nu_j$, with $j\in J$ satisfying 
\begin{small}
\begin{equation}\label{ineq:thJ} 
\max_{j\in J}\left\{\frac{k_{\rm{cat}_j}}{\nu_j}\right\}< \left(\frac{S_{tot}}{E_{tot}}-1\right)F_{tot}.
\end{equation}
\end{small}%
Then, there exist positive constants $B_1,\dots,B_n$ such that for any choice of positive constants $\lambda_0,\dots,\lambda_{n-1}$ satisfying
\begin{small}
\begin{equation}\label{ineq:gammaallintermediates}
\frac{\lambda_j}{\lambda_{j-1}}<B_j \ \text{ for } j=1,\dots,n-1,\quad \frac{1}{\lambda_{n-1}}<B_n, 
\end{equation}
\end{small}%
rescaling of the given parameters $k_{\rm{on}_j}$ by $\lambda_{j}\, k_{\rm{on}_j}$, for $j\in J$ and $\tau_j$ by
$\lambda_{j}\tau_j$ if $j\notin J$ gives rise to a system with at least $1+2|J|$ positive steady states.
\end{thm}

The following immediate Corollary of Theorem~\ref{th:J} implies that  we can obtain a region in parameters space with  
$[\frac n 2]$ intermediates, where the associated system has $2 [\frac n 2]+1$ positive steady states.

\begin{corollary}\label{cor:1/4}
Let $n\geq 1$, and consider the network $G_{J_n}$ as in~\eqref{eq:net_phospho_n_lessintermediatesEside}, 
with $J_n$ as in~\eqref{Jn}. Assume moreover that  
  \begin{small}
\[S_{tot}>E_{tot}.\] 
\end{small}%
Then, there is a multistationarity region in the space of all parameters for which the network $G_{J_n}$ 
admits $2[\frac n 2]+1$ steady states
(with fixed linear conservation constants corresponding to the coordinates of a vector of parameters in this region), 
described in the statement of Theorem~\ref{th:J}.
\end{corollary}

We explain in Section~\ref{sec:R}  the computational approach to find new regions of multistationarity. This is derived from
the framework in~\cite{BDG1} that we used to get the previous results. We find new regions of multistationarity for
the cases $n=2,3, 4$ and $5$, after some manual organization of the automatic computations. Our computer aided results are summarized
in Propositions~\ref{prop:maple2},~\ref{prop:maple3},~\ref{prop:maple3b},~\ref{prop:maple4},~\ref{prop:maple5} and~\ref{prop:maple5b}.

\smallskip

The paper is organized as follows. In Section~\ref{sec:upperandmore} we prove Proposition~\ref{thm:upperboundGJ}
and we show how to lift regions of multistationarity from the reduced system $G_J$ to the complete sequential $n$-site phosphorylation system. 
We also show in Lemma~\ref{lemma:0} that even with a single intermediate $Y_0$ it is possible to make a choice of all parameters such that the system has
$2 [\frac n 2]+1$ positive steady states.  This result has been independently found by Elisenda Feliu (personal communication).
This says that while Corollary~\ref{cor:1/4} is optimal,  the regions obtained for any subset $J$ with indexes of the same parity of $n$
in Theorem~\ref{th:J} properly contained in $J_n$, only ensure $2 |J|+1$ positive steady states. However, note that we are able to describe open regions in parameter space and Lemma~\ref{lemma:0} only allows us to get choices of parameters.

In Section~\ref{sec:background} we briefly recall the framework in~\cite{BDG1}, which is the basis of our approach. 
In Section~\ref{sec:proofs} we give the proofs of Theorems ~\ref{th:allintermediatesEside} and~\ref{th:J}. Finally, as we mentioned above, 
we explain  in Section \ref{sec:R} how to implement the computational approach to find new regions of multistationarity. 
We end with a discussion where we further compare our detailed results with previous results in the literature.

\section{Upper bounds and extension of multistationarity}\label{sec:upperandmore}

In this section we collect three results that complete our approach to describe multistationarity regions giving lower bounds for the number of
positive steady states with fixed linear conservation relations for the systems $G_J$ in Definition~\ref{def:G_J}
for any $J\subset I_n$. 
We first show a positive parametrization of the concentrations at steady state of the species of the systems $G_J$,
which allows us to translate the equations defining the steady states in all the concentrations plus the linear conservation
relations into a system of two equations in two variables.
In Proposition~\ref{thm:upperboundGJ}, we prove that any
 mass-action system  $G_J$, has at most $2 [\frac n 2]+1$ positive steady states. 
Lemma~\ref{lemma:0}  shows that having a single intermediate is enough to get that number of
positive steady states, for particular choices of the reaction rate constants. However, this computation by reduction to
the univariate case is not systematic as the general approach from~\cite{BDG1} that we use to describe multistationarity
regions in Theorems~\ref{th:allintermediatesEside} and~\ref{th:J}, which can be applied to study other quite different mechanisms.
It is known that if a system $G_J$ has $m$ nondegenerate positive steady states for a subset $J \subset I_n$, then it is possible to find parameters 
for the whole $n$-site phosphorylation system that also give at least $m$ positive steady states (see~\cite{fw13}). In Theorem~\ref{th:lifting},
we give precise conditions on the rate constants to lift the regions of multistationarity for the reduced networks to regions of multitationarity
with $2[\frac n 2]+1$ positive steady states (with fixed linear conservation constants) of the complete $n$-sequential phosporylation cycle.

\subsection{Parametrizing the steady states}\label{ssec:param}

The following lemma gives a positive parame\-trization of the concentration of the species at steady state for
the systems $G_J$, 
in terms of the concentrations of the unphosphorylated substrate $S_0$ and the kinase $E$. 
It is a direct application of the general procedure presented in Theorem~3.5  in~\cite{aliciaMer}, and
generalizes the parametrization given in Section 4 of~\cite{BDG1}.

\begin{lemma} \label{lemma:par}
Given $n\geq 1$ and a subset $J\subset I_n$, consider the system $G_J$ as in Definition~\ref{def:G_J}.
Denote  for each $j\in J$
\begin{small}
\begin{equation}\label{eq:Kj}
K_j=\frac{k_{\rm{on}_j}}{k_{\rm{off}_j}+k_{\rm{cat}_j}},  \quad \tau_j=k_{\rm{cat}_j}K_j,
\end{equation}
\end{small}%
set $T_{-1}=1$, and  for any $i=0,\dots,n-1$:
\begin{small}
\begin{equation} \label{eq:Ti}
T_i=\prod^{i}_{j=0}\frac{\tau_j}{\nu_j}.
\end{equation}
\end{small}%
Then, the parametrization of the concentrations of the species
at a steady state in terms of $s_0$ and $e$ is equal to:
\begin{small}
\begin{equation}\label{parametrizationGJ}
 s_i =\frac{T_{i-1}}{(F_{tot})^i} \, {s_0e^i},\ i=1,\dots,n, \qquad
 y_j = \frac{K_j\, T_{j-1}}{(F_{tot})^j} \, s_0e^{j+1},\ j\in J, 
\end{equation}
\end{small}
\end{lemma}
The inverses  $K_j^{-1}$ of the constants $K_j$ in~\eqref{eq:Kj} in the statement of
 Lemma~\ref{lemma:par} are usually called \textit{Michaelis-Menten constants}.

\smallskip

Let $n\geq 1$ and a subset $J\subset I_n$. If we substitute the monomial parametrization of the concentration 
of the species at steady state~\eqref{parametrizationGJ} into the conservation laws, we obtain a system of two 
equations in two variables $s_0$ and $e$. We have: 
\begin{small}
\begin{eqnarray}
s_0 + \sum_{j\in J} \left(\frac{T_{j}}{(F_{tot})^{j+1}} + \frac{K_j\, T_{j-1}}{(F_{tot})^j}\right)\, {s_0e^{j+1}}
 + \sum_{j\notin J} \frac{T_{j}}{(F_{tot})^{j+1}} \, {s_0e^{j+1}} - S_{tot}=0,\label{eq:systemG_J}\\
e + \sum_{j\in J} \frac{K_j\, T_{j-1}}{(F_{tot})^j}\, {s_0e^{j+1}} - E_{tot}=0.\nonumber
\end{eqnarray}
\end{small}

We can write system \eqref{eq:systemG_J} in matricial form:
\begin{small}
\begin{equation}\label{eq:C}
C \begin{pmatrix} s_0 & e & s_0e & s_0e^2 & \dots & s_0e^{n} & 1
\end{pmatrix}^t=0,
\end{equation}
\end{small}
where $C\in\R^{2\times (n+3)}$ is the matrix of coefficients:
\begin{small}
\begin{equation}\label{eq:matrixC} 
C=\begin{pmatrix} 1 & 0 & \frac{T_{0}}{F_{tot}} + \beta_0 & \frac{T_{1}}{(F_{tot})^2} +
 \beta_1 & \dots &  \frac{T_{n-1}}{(F_{tot})^n} + \beta_{n-1} & -S_{tot}\\
0 & 1 & \beta_0 & \beta_1 & \dots & \beta_{n-1} & -E_{tot}
\end{pmatrix},
\end{equation}
\end{small}
with 
\begin{equation}\label{eq:beta}
\small
\beta_j=\frac{K_j\, T_{j-1}}{(F_{tot})^j}  \text{  for } j\in J,   \text{ and } \beta_j=0  \text{ if } j\notin J.
\end{equation}

Note that if we order the variables $s_0$, $e$, the support of the system  (that is, the exponents of
the monomials that occur) is the following set $\mathcal A$:
\begin{equation}\label{eq:support}
\mathcal{A}=\{(1,0),(0,1),(1,1),(1,2),\dots,(1,n),(0,0)\}\subset \Z^2,
\end{equation}
independently of the choice of $J\subset I_n$.

\subsection{Upper bounds on the number of positive steady states}\label{ssec:complex}

We first recall Kushnirenko Theorem, a fundamental result about sparse systems of polynomial equations,  which gives a bound on the
number of complex solutions with nonzero coordinates.
Given a finite point configuration $\mathcal{A}\subset\Z^d$, denote by $\ch{(\A)}$ the convex hull of $\mathcal{A}$. 
We write $\vol$ to denote Euclidean volume, and set $\C^{*}=\C\setminus \{0\}$. 

\smallskip

\noindent{\bf Kushnirenko Theorem~\cite{K}:} 
Given a finite point configuration $\mathcal{A}\subset\Z^d$,
a sparse system of $d$ Laurent polynomials in $d$ variables with 
support $\mathcal{A}$ has at most $d!\,\vol(\ch{(\A)})$ isolated solutions in 
$(\C^{*})^d$ (and exactly this number if the polynomials have {\em generic} coefficients.)

\medskip

We also recall the classical Descartes rule of signs.

\smallskip

\noindent{\bf Descartes rule of signs:} 
Let $p(x)=c_0+c_1x+\dots+c_mx^m$ be a nonzero univariate real polynomial with $r$ positive real roots counted with multiplicity. Denote by $s$ the number of sign variations in the ordered sequence of the signs 
$\sign(c_0),\dots,\sign(c_m)$ of the coefficients, i.e., discard the $0$'s in this sequence and then count 
the number of times two consecutive signs differ. Then $r\leq s$ and $r$ and $s$ have the same parity, which is
even if $c_0 c_m >0$ and odd if $c_0 c_m <0.$

 \medskip
 
We then have that $2 [\frac n 2]+1$ is an upper bound for the number of positive real solutions of the system of equations defining the steady states 
of any system $G_J$ in Definition~\ref{def:G_J}:

\begin{proposition}\label{thm:upperboundGJ} For any choice of rate constants and total conservation constants, 
the dynamical system $G_J$ associated with any subset $J \subset  I_n$ has at most $2[\frac n 2]+1$ isolated positive steady states. 
In fact, the polynomial system of equations defining the steady states of $G_J$
can have at most $n+1$ isolated solutions in $(\C^{*})^d$.
\end{proposition}

\begin{proof} The number of positive steady states of the subnetwork $G_J$ is the number of positive solutions of the 
sparse system~\eqref{eq:systemG_J} of two equations and two variables. The support of the system is \eqref{eq:support} whose convex hull is a triangle with Euclidean volume equal to $\frac{n+1}{2}$. By Kushnirenko Theorem,
the number of isolated solutions in  $(\C^{*})^2$ is at most $2!\frac{(n+1)}{2}=n+1$. In particular, the number of 
isolated positive solutions is at most $n+1$.

Moreover, when all positive solutions are nondegenerate,  their number is necessarily 
 odd by Corollary~2 in~\cite{CFMW17}, which is based on the notion of Brouwer's degree. Indeed, in our case, 
we can bypass the condition of nondegeneracy because we can use Descartes rule of signs in one variable.
In fact, from the first equation of system~\eqref{eq:systemG_J}, we can write:
\begin{small}
\begin{equation}\label{eq:s0}
s_0=\frac{S_{tot}}{p(e)},
\end{equation}
\end{small}%
where $p(e)$ is the following polynomial of degree $n$ on the variable $e$:
{\small
\begin{equation}
p(e):=1+\sum_{i=0}^{n-1}(\alpha_i+\beta_i)e^{i+1},
\end{equation}}%
with 
\begin{equation}\label{eq:alpha}
\small
\alpha_i=\frac{T_{i}}{(F_{tot})^{i+1}}, \quad i=0,\dots,n-1,
\end{equation}
 and  $\beta_{i}=
\frac{K_j\, T_{j-1}}{(F_{tot})^j}$ if $j\in J$ or $\beta_{j}=0$ if $j\notin J$ were defined in~\eqref{eq:beta}. 
Note that for any  $e>0$ it holds that $p(e)>0$, and so $s_0 >0$.
If we replace~\eqref{eq:s0} in the second equation of~\eqref{eq:systemG_J}, we have:
{\small 
\begin{equation}\label{polynomial-e}
e + \sum_{i=0}^{n-1} \beta_i \frac{S_{tot}}{p(e)} e^{i+1} - E_{tot}=0.
\end{equation}}%
 The number of positive solutions of the equation \eqref{polynomial-e} is the same if we multiply by $p(e)$:
 {\small
\begin{equation}\label{eq:polyq}
q(e):=e\,p(e) + \sum_{i=0}^{n-1} \beta_i{S_{tot}}e^{i+1} - E_{tot}\, p(e)=0.
\end{equation}}%
This last polynomial $q$ has degree $n+1$, with leading coefficient equal to $\alpha_{n-1}+\beta_{n-1}>0$ and 
constant coefficient equal to $-E_{tot}<0$. The sign variation of the coefficients of $q$ has the same parity as the 
sign variation of the leading coefficient and the constant coefficient, which is one. So, by Descartes rule of signs, 
as the sign variation is odd, the number of positive solutions is also odd.
\end{proof}

\subsection{One intermediate is enough in order to obtain  $2[\frac n 2]+1$ positive steady states}

As we mentioned in the introduction, the following result has been independently found by Elisenda Feliu (personal communication).

\begin{lemma} \label{lemma:0}
If $J=\{0\}$, then there exists parameter values such that the system $G_J$ admits  $2[\frac n 2]+1$ positive steady states.
\end{lemma}

\begin{proof}
Assume $n$ is even, then $n+1$ is odd. As we did in the proof of Proposition~\ref{thm:upperboundGJ}, 
the positive solutions of the system~\eqref{eq:systemG_J} are in bijection with the positive 
solutions of the polynomial $q(e)$ in~\eqref{eq:polyq}. Here $\beta_0=K_0$ and $\beta_i=0$ for $i\neq 0$. 
We will consider the polynomial $\tilde q(e):=\frac{q(e)}{E_{tot}}$, with constant coefficient equal to $-1$.

Consider any polynomial of degree $n+1$
{\small
\begin{equation}\label{eq:polyn+1} 
c_{n+1}e^{n+1}+c_ne^n+\dots+c_1e-1,
\end{equation}}%
with $n+1$ distinct positive roots, and with constant term equal to $-1$. Then, we have that $c_{i}(-1)^{i+1}>0$, and in
particular, $c_{n+1}>0$.

Our goal is to find reaction rate constants and total conservation constants such that the polynomial~\eqref{eq:polyn+1} 
coincides with the polynomial $\tilde q(e)$. Comparing the coefficient of $e^i$, for $i=1,\dots,n+1$ in both polynomials, we need 
to have:
{\small
\begin{eqnarray}
\frac{\alpha_{n-1}}{E_{tot}}&=&c_{n+1},\nonumber \\
\frac{\alpha_{i-2}}{E_{tot}} - \alpha_{i-1 } &=& c_{i}, \, \text{ for } i=3,\dots,n, \label{eq:coefficients}\\
\frac{\alpha_{0}+K_0}{E_{tot}} - \alpha_{1 } &=& c_{2}, \nonumber \\
\frac{1+S_{tot}K_0}{E_{tot}} - (\alpha_{0}+K_0) &=& c_1.\nonumber
\end{eqnarray}   }%
Keep in mind that the values of $c_i$ are given. We may solve \eqref{eq:coefficients} in terms of $E_{tot}$ and of the $c_i,i=1,\dots,n+1:$
{\small 
\begin{eqnarray}
\alpha_{n-1-k}&=& E_{tot}\sum_{i=0}^k c_{n+1-i}\, (E_{tot})^{k-i},   \, \text{ for each } k =0,1,\dots,n-2,\nonumber\\
\alpha_0+K_0 &=& E_{tot}\sum_{i=0}^{n-1} c_{n+1-i}\, (E_{tot})^{n-1-i},  \label{eq:coefficients2} \\
1+S_{tot}K_0&=& E_{tot}\sum_{i=0}^n c_{n+1-i}\, (E_{tot})^{n-i}.   \nonumber
\end{eqnarray} }%
 Note that if we take any value for $E_{tot}>0$, 
then the values of $\alpha_i$ for $i=1,\dots,n-1$, $\alpha_0 + K_0$ and $S_{tot}K_0$ are completely determined. 
So, we find an appropriate value of $E_{tot}$ such that the previous values $\alpha_i$, $K_0$ and 
$S_{tot}$ are all positive. For that, we choose $K_0=1$, and we take $E_{tot}$ large enough such that 
{\small
\begin{eqnarray}
\sum_{i=0}^k c_{n+1-i}\, (E_{tot})^{k-i} & > & 0, \, \text{ for each } k \in \{0,1,\dots,n-2\} \text{ with } k \text{ odd}, \nonumber \\
E_{tot}\sum_{i=0}^{n-1} c_{n+1-i}\, (E_{tot})^{n-1-i} & > &1,\quad\nonumber
E_{tot}\sum_{i=0}^{n-1} c_{n+1-i}\, (E_{tot})^{n-i} > 1.
\end{eqnarray} }%
This is possible since $c_{n+1}>0$ and that imposing the first condition just on $k$ odd 
is enough to ensure that it holds for all $k\in I_{n-1}$ as well because of the signs of the $c_i.$ 
With these conditions, and using the equations \eqref{eq:coefficients2}, 
the values of $\alpha_i$ for each $i=0,\dots,n-1$ and $S_{tot}$ are determined and are all positive.

Now, it remains to show that we can choose reaction rate constants such that the values of $\alpha_i, i=0,\dots,n-1$ are the given ones. 
Recall that
$\alpha_i=\frac{T_i}{(F_{tot})^{i+1}}$, where $T_i=\prod^{i}_{j=0}\frac{\tau_j}{\nu_j}$ for $i=0,\dots,n-1$ 
and $T_{-1}=1$, where $\tau_0=k_{\rm{cat}_0}K_0=k_{\rm{cat}_0}$ (we have chosen $K_0=1$). 
Take for example $F_{tot}=1$, $k_{\rm{on}_0}=2$, $k_{\rm{off}_0}=1$ and $k_{\rm{cat}_0}=1$ 
(to obtain $K_0=1$). Then, $\tau_0=1$, so we take $\nu_0=\frac{1}{\alpha_0}$. As 
$\alpha_{i+1}=\alpha_{i}\frac{\tau_{i+1}}{\nu_{i+1}}$, for $i=0,\dots,n-2$, we can choose any positive values of $\tau_{i+1}, \nu_{i+1}$
such that $\frac{\tau_{i+1}}{\nu_{i+1}}= \frac{\alpha_{i+1}}{\alpha{i}}$, and we are done.

When $n$ is odd, with a similar argument, we can find reaction rate constants and total conservation 
constants such that the polynomial~$\tilde q(e)$ gives a polynomial like~\eqref{eq:polyn+1}
 (but with $n$ distinct positive roots and one negative root).
\end{proof}

\subsection{Lifting regions of multistationarity}
Multistationarity of any of the subsystems $G_J$ can be extended to the full $n$-site phosphorylation system (for instance,
by Theorem 4 in~\cite{fw13}).
We give a precise statement of this result in Theorem~\ref{th:lifting}. 

Consider the full $n$-site phosphorylation network \eqref{eq:nsite}, with a given vector of reaction rate constants 
$\kappa\in\R^{6n}$:
{\small \[\kappa=(k_{\rm{on}_0}, k_{\rm{off}_0}, k_{\rm{cat}_0}, \dots, k_{\rm{on}_{n-1}}, k_{\rm{off}_{n-1}}, 
k_{\rm{cat}_{n-1}}, \ell_{\rm{on}_0}, \ell_{\rm{off}_0}, \ell_{\rm{cat}_0}, \dots, \ell_{\rm{on}_{n-1}}, \ell_{\rm{off}_{n-1}}, \ell_{\rm{cat}_{n-1}}).\]}%
We define the following rational functions of $\kappa$:
{\small
\begin{equation}\label{eq:constantsG_Jkappa}
\tau_j(\kappa)= k_{\rm{cat}_j} \, \mu_j(\kappa) \text{ if } j\notin J \, \text{ and } \, \nu_j(\kappa)=
\ell_{\rm{cat}_j}\, \eta_j(\kappa) \text{ for } j=0,\dots,n-1, 
\end{equation}}%
where $\mu_j(\kappa)$ and $\eta_j(\kappa)$ are in turn the following rational functions:
{\small
\begin{equation}
\mu_j(\kappa)=\frac{k_{\rm{on}_j}}{k_{\rm{off}_j}+k_{\rm{cat}_j}} \text{ if } j\notin J \, \text{ and } \, \eta_j(\kappa)=
\frac{\ell_{\rm{on}_j}}{\ell_{\rm{off}_j}+\ell_{\rm{cat}_j}} \text{ for } j=0,\dots,n-1.
\end{equation}}%
We denote by $\varphi\colon\R_{>0}^{6n}\to \R_{>0}^{2n+2|J|}$ the function that takes $\kappa$ 
and gives a vector of (positive) reaction rate constants with the following order: 
first, the constants $k_{\rm{on}_j}, k_{\rm{off}_j}, k_{\rm{cat}_j},j\in J$, then $\tau(\kappa), j\notin J$, and then $\nu_j(\kappa), j=0,\dots, n-1$.

Given a subset $J\subset I_n$ and a vector of reaction rate constants $\kappa\in\R_{>0}^{6n}$, we consider 
the subnetwork $G_J^{\varphi(\kappa)}$ as in Definition~\ref{def:G_J}, with rate constants $\varphi(\kappa)$:
{\small
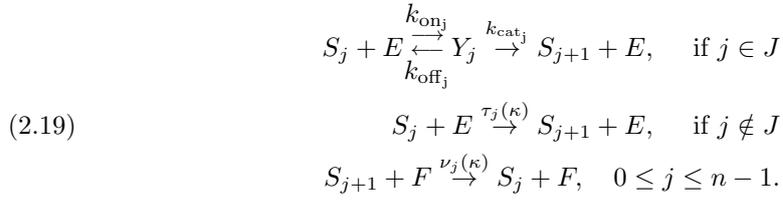
 \begin{eqnarray}   
 S_j+E 
 \arrowschem{k_{\rm{on}_j}}{k_{\rm{off}_j}} Y_j
 \stackrel{k_{\rm{cat}_j}}{\rightarrow} S_{j+1}+E,\quad \text{ if } j\in J\nonumber \\
 S_j+E   \stackrel{\tau_j(\kappa)}{\rightarrow} S_{j+1}+E,\quad \text{ if } j\notin J \label{eq:net_phospho_n_lessintermediatesEside2}\\
 S_{j+1}+F 
 \stackrel{\nu_{j}(\kappa)}{\rightarrow} S_{j}+F,\quad 0\leq j \leq n-1. \nonumber
\end{eqnarray}}

Applying Theorem 6.4 from \cite{DGPR}, which is built on Theorem 4 in \cite{fw13}, we get the following lifting result. 

\begin{thm}\label{th:lifting} Consider the full $n$-site phosphorylation network \eqref{eq:nsite} with fixed reaction 
rate constants $\kappa^0$ and the network $G_J^{\varphi(\kappa^0)}$, both with
total conservation amounts $S_{tot}$, $E_{tot}$, $F_{tot}>0$. Suppose that system $G_J^{\varphi(\kappa^0)}$ 
admits $m$  nondegenerate positive steady states. 

Then, there exists $\varepsilon_0 >0$ such that for any choice of rate constants $\kappa$ such that $\varphi(\kappa)=\varphi(\kappa^0)$ and
{\small
\begin{equation}\label{eq:mueta}
\max_{j\notin J}|\mu_j(\kappa)|, \quad \max_{j\in I_n}|\eta_j(\kappa)| \, < \, \varepsilon_0, 
\end{equation}}%
the $n$-site sequential phosphorylation system admits $m$ positive nondegenerate  steady states in the stoichiometric compatibility class 
defined by $S_{tot}$, $E_{tot}$ and $F_{tot}>0$.
Moreover, the set of rate constants $\kappa$ verifying $\varphi(\kappa)=\varphi(\kappa^0)$ and~\eqref{eq:mueta} is nonempty.
\end{thm}

This last result allows us to obtain multistationary regions for the complete $n$-site phosphorylation system, combining the
 conditions on the parameters given in Theorem~\ref{th:allintermediatesEside} and Theorem~\ref{th:J}, 
 with conditions \eqref{eq:mueta} of Theorem~\ref{th:lifting}.  In particular, let $J_n \subset I_n$ as in~\eqref{Jn}. By lifting a multistationarity region
 for the system $G_{J_n}$ in Corollary~\ref{cor:1/4}, we get a multistationarity region of parameters of the $n$-site phosphorylation
 cycle with $2 [\frac n 2]+1$ positive steady states in the same stoichiometric compatibility class.

\section{Positive solutions of sparse polynomial systems}\label{sec:background}

As we saw in Subsection~\ref{ssec:param}, the steady states of the systems $G_J$ correspond to the positive solutions
of the sparse polynomial system~\eqref{eq:systemG_J} in two variables. 
In this section, we briefly recall the general setting  from \cite{BDG1,bihan} to find lower bounds for
sparse polynomial systems, that we will use in the proof of
Theorems~\ref{th:allintermediatesEside} and \ref{th:J} in Section~\ref{sec:proofs} and also in Section~\ref{sec:R}.
For detailed examples of this approach, we refer the reader to Section 2 in~\cite{cascades}.

\medskip

We consider a fixed finite point configuration $\mathcal{A} =
  \{a_1, \dots, a_n\}\subset \Z^d$, with $n\geq d+2$. A sparse \textit{polynomial system} of $d$ Laurent 
  polynomial equations with \textit{support} $\mathcal{A}$ is a system 
 $f_{1}(x)=\dots=f_{d}(x)=0$ in $d$ variables $x=(x_1, \dots, x_d)$, with
 \begin{equation}\label{systemf}
 f_{i}(x)=\sum_{j=1}^n c_{ij} \, x^{a_j} \in \R[x_1,\dots,x_d], \ i=1,\dots,d,
 \end{equation}%
where the exponents belong to $\mathcal{A}$. We call $C = (c_{ij}) \in \R^{d \times n}$ the 
  \textit{coefficient matrix} of the system and we assume 
 that no column of $C$ is identically zero.   Recall that
a zero of~\eqref{systemf} is nondegenerate when it is not a zero of the Jacobian of $f_{1}, \dots, f_{d}$.
   
Our method to obtain a lower bound on the number of positive steady states, based on \cite{BDG1,bihan}, is to restrict our polynomial 
system \eqref{systemf} to subsystems which have a positive solution and then extend 
these solutions to the total system, via a deformation of the coefficients. 
The first step is then to find conditions in the coefficient matrix $C$ that guarantee a positive solution to each of the subsystems. 
The hypothesis of having subsystems supported on {\em simplices} of a {\em regular subdivision} of $\mathcal{A}$ is the key to
the existence of a on open set in the space of coefficients where all
these solutions can be extended. We recall below only the  concepts that we need for our work. 
A comprehensive treatment of this subject can be found in \cite{booktriang}. Following Section 3 in~\cite{bihan}, we define:

 \begin{definition} Given a matrix $M\in\R^{d\times (d+1)}$,  we denote
 by  $\minor(M,i)$ the determinant of the square matrix obtained by
  removing the $i$-th column of $M$. The matrix $M$ is called positively spanning 
 if all the values $(-1)^i\minor(M,i)$, for $i=1,\dots,d+1$, are nonzero and have the same sign.
 \end{definition}
 Equivalently, a matrix $M$ is positively spanning if all the coordinates of any nonzero vector in ${\rm ker}(M)$
are non-zero and have the same sign. 

A \textit{$d$-simplex} with vertices in $\mathcal{A}$ is a subset of $d+1$ points of $\mathcal{A}$ which is affinely independent. 
By Proposition~3.3 in \cite{bihan}, a system of $d$ polynomial equations in $d$ 
variables with a $d$-simplex as support, has one non-degenerate positive solution if and only if its $d\times (d+1)$ 
matrix of coefficients is positively spanning. We further define, following~\cite{bihan}: 

\begin{definition} \label{def:dec}
 Let $C\in\R^{d\times n}$. We say that a $d$-simplex 
 $\Delta=\{a_{i_1},\dots,a_{i_{d+1}}\}$ is positively decorated by $C$ if the $d\times(d+1)$ submatrix of $C$ with 
 columns $\{i_1,\dots,i_{d+1}\}$ is positively spanning.
 \end{definition}

Given a fixed finite point configuration $\mathcal{A}$, take a {\em height function} $h:\A \to \R$, $h=(h(a_1), \dots, h(a_n))$. Consider the \textit{lower
convex hull} $\mathcal L$ of the $n$ lifted
points $(a_j,h(a_j))\in \R^{d+1}, j=1, \dots, n$ (see Figure~\ref{fig:Regular triangulation}). Project to
 $\R^d$ the subsets of points  in each of the faces of $\mathcal L$ (that it, points where the affine linear function which defines the lower 
 face is minimized). 
 These subsets define a \emph{regular}
subdivision  of $\mathcal{A}$ induced by $h$. When the height 
vector $h$ is generic, the regular subdivision is a
regular triangulation, in which all the subsets are simplices. 

 \begin{center}
	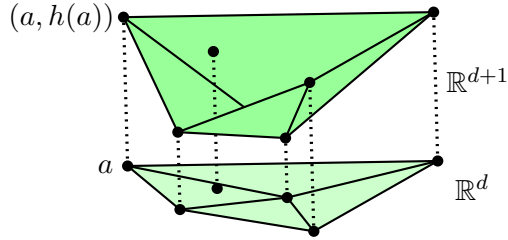
\begin{figure}[h]
				\centering
			\begin{tikzpicture}%
			[x={(1.295303cm, 0.021152cm)},
			y={(0.905733cm, 0.517961cm)},
			z={(-0.024093cm, 0.8770689cm)},
			scale=0.900000,
			back/.style={dotted,very thick},
			edge/.style={color=black, thick},
			facet/.style={fill=blue!95!black,fill opacity=0.800000},
			vertex/.style={inner sep=1.2pt,circle,draw=black,fill=black,thick,anchor=base}]
			%
			%
			\coordinate (0.00000, 0.00000, 0.00000) at (0.00000, 0.00000, 0.00000);
			\coordinate (3.00000, 0.00000, 0.00000) at (3.00000, 0.00000, 0.00000);
			\coordinate (3.00000, 2.00000, 0.00000) at (3.00000, 2.00000, 0.00000);
			\coordinate (0.00000, 2.00000, 0.00000) at (0.00000, 2.00000, 0.00000);
			\coordinate (0.00000, 0.00000, 2.50000) at (0.00000, 0.00000, 2.50000);
			\coordinate (3.00000, 0.00000, 2.50000) at (3.00000, 0.00000, 2.50000);
			\coordinate (3.00000, 2.00000, 2.50000) at (3.00000, 2.00000, 2.50000);
			\coordinate (0.00000, 2.00000, 2.50000) at (0.00000, 2.00000, 2.50000);
			\coordinate (2.00000, 1.00000, 0.00000) at (2.00000, 1.00000, 0.00000);
			\coordinate (1.00000, 0.70000, 0.00000) at (1.00000, 0.70000, 0.00000);
			\coordinate (1.00000, 1.30000, 0.00000) at (1.00000, 1.30000, 0.00000);
			\coordinate (2.00000, 1.00000, 1.00000) at (2.00000, 1.00000, 1.00000);
			\coordinate (1.00000, 0.70000, 1.30000) at (1.00000, 0.70000, 1.30000);
			\coordinate (1.00000, 1.30000, 2.30000) at (1.00000, 1.30000, 2.30000);
			
			

			\draw[edge,back](-0.500000, 2.00000, 0.00000) -- (-0.500000, 2.00000, 2.50000);
			\draw[edge,back](3.00000, 2.00000, 0.00000) -- (3.00000, 2.00000, 2.50000);
			\fill[facet, fill=green!40,fill opacity=1.000000] (-0.500000, 2.00000, 2.50000) -- (1.00000, 0.70000, 1.30000) -- (2.00000, 1.00000, 1.00000) -- cycle {};
			\fill[facet, fill=green!40,fill opacity=1.000000] (-0.500000, 2.00000, 2.50000) -- (3.00000, 2.00000, 2.50000) -- (2.00000, 1.00000, 1.00000) -- cycle {};
			\draw[edge] (-0.500000, 2.00000, 2.50000) -- (2.00000, 1.00000, 1.00000);
			\fill[facet, fill=green!20,fill opacity=1.000000] (-0.500000, 2.00000, 0.00000) -- (1.00000, 0.70000, 0.00000) -- (2.00000, 1.00000, 0.00000) -- cycle {};
			\fill[facet, fill=green!20,fill opacity=1.000000] (-0.500000, 2.00000, 0.00000) -- (3.00000, 2.00000, 0.00000) -- (2.00000, 1.00000, 0.00000) -- cycle {};
			\draw[edge,back](1.00000, 1.30000, 0.00000) -- (1.00000, 1.30000, 2.30000);
			\draw[edge] (1.00000, 0.70000, 1.30000) -- (-0.500000, 2.00000, 2.50000);
			\fill[facet, fill=green!40,fill opacity=1.000000] (2.00000, 1.00000, 1.00000) -- (3.00000, 2.00000, 2.50000) -- (3.00000, 0.00000, 2.50000) -- cycle {};
			\fill[facet, fill=green!40,fill opacity=1.000000] (3.00000, 0.00000, 2.50000) -- (1.00000, 0.70000, 1.30000) -- (2.00000, 1.00000, 1.00000) -- cycle {};

			\fill[facet, fill=green!20,fill opacity=1.000000] (2.00000, 1.00000, 0.00000) -- (3.00000, 2.00000, 0.00000) -- (3.00000, 0.00000, 0.00000) -- cycle {};
			\fill[facet, fill=green!20,fill opacity=1.000000] (3.00000, 0.00000, 0.00000) -- (1.00000, 0.70000, 0.00000) -- (2.00000, 1.00000, 0.00000) -- cycle {};
			\draw[edge] (3.00000, 0.00000, 2.50000) -- (1.00000, 0.70000, 1.30000);
			\draw[edge] (2.00000, 1.00000, 1.00000) -- (1.00000, 0.70000, 1.30000);
			\draw[edge] (3.00000, 0.00000, 2.50000) -- (2.00000, 1.00000, 1.00000);
			\draw[edge] (3.00000, 0.00000, 0.00000) -- (3.00000, 2.00000, 0.00000);
			\draw[edge] (-0.500000, 2.00000, 0.00000) -- (3.00000, 2.00000, 0.00000);
			\draw[edge] (3.00000, 0.00000, 0.00000) -- (2.00000, 1.00000, 0.00000);
			\draw[edge] (3.00000, 2.00000, 0.00000) -- (2.00000, 1.00000, 0.00000);
			\draw[edge] (-0.500000, 2.00000, 0.00000) -- (2.00000, 1.00000, 0.00000);
			\draw[edge] (1.00000, 0.70000, 0.00000) -- (2.00000, 1.00000, 0.00000);
			\draw[edge] (1.00000, 0.70000, 0.00000) -- (3.00000, 0.00000, 0.00000);
			\draw[edge] (1.00000, 0.70000, 0.00000) -- (-0.500000, 2.00000, 0.00000);
			\draw[edge,back] (3.00000, 0.00000, 2.50000) -- (3.00000, 0.00000, 0.00000);
			\draw[edge,back](2.00000, 1.00000, 0.00000) -- (2.00000, 1.00000, 1.00000);
			\draw[edge,back] (1.00000, 0.70000, 0.00000) -- (1.00000, 0.70000, 1.30000);
			\draw[edge] (3.00000, 0.00000, 2.50000) -- (3.00000, 2.00000, 2.50000);
			\draw[edge] (3.00000, 2.00000, 2.50000) -- (2.00000, 1.00000, 1.00000);
			\draw[edge] (3.00000, 2.00000, 2.50000) -- (-0.500000, 2.00000, 2.50000);
			\node at (-0.75,2,0) {$a$};
			\node[vertex] at (-0.500000, 2.00000, 0.00000)     {};
			\node[vertex] at (3.00000, 0.00000, 0.00000)     {};
			\node[vertex] at (3.00000, 2.00000, 0.00000)     {};
			\node at (4.8,0,0.7) {$\R^d$};
			\node at (4.9,0,2.5) {$\R^{d+1}$};
			\node[vertex] at (-0.500000, 2.00000, 2.50000)     {};
			\node[vertex] at (3.00000, 0.00000, 2.50000)     {};
			\node at (-1.2,2,2.5) {$(a,h(a))$};
			\node[vertex] at (3.00000, 2.00000, 2.50000)     {};
			\node[vertex] at (2.00000, 1.00000, 0.00000)   	 {};
			\node[vertex] at (1.00000, 1.30000, 0.00000) 	 {};
			\node[vertex] at (1.00000, 0.70000, 0.00000)	 {};
			\node[vertex] at (2.00000, 1.00000, 1.00000)   	 {};
			\node[vertex] at (1.00000, 0.70000, 1.30000) 	 {};
			\node[vertex] at (1.00000, 1.30000, 2.30000)	 {};

			\end{tikzpicture}
			\caption{Regular triangulation.} \label{fig:Regular triangulation}
			\vspace{\baselineskip}
	\end{figure}
\end{center}

Note that the set of all height vectors inducing 
a regular subdivision of $\mathcal{A}$ that contains certain $d$-simplices $\Delta_1, \dots, \Delta_p$  is defined by 
a finite number of linear inequalities. Thus, this set is a finitely generated convex cone ${\mathcal C}_{\Delta_1,\dots, \Delta_p}$ 
in $\R^n$ with apex at the origin. In particular, the set of all height vectors inducing a regular subdivision $\Gamma$ of $\mathcal{A}$ 
is a finitely generated convex cone in $\R^n$, which we call $\mathcal{C}_{\Gamma}$. In particular, if the simplices $\Delta_1, \dots, \Delta_p$ 
completely determine the regular subdivision $\Gamma$, 
the cone ${\mathcal C}_{\Delta_1,\dots, \Delta_p}$ is equal to the cone ${\mathcal C}_{\Gamma}$.

Given $\mathcal{A}\subset\Z^{d}$, consider a  system of sparse real polynomials $f_1 = f_2 = \dots = f_d=0$ as in~\eqref{systemf} with support $\mathcal A$. 
Let $\Gamma$ be a regular subdivision of $\mathcal A$ and $h$ any height function 
that induces $\Gamma$. We define the following family of real polynomial systems parametrized by a positive real
number $t>0$: 
 \begin{equation}\label{systemwitht}
 \sum_{j=1}^n c_{ij}\, t^{h(a_j)} \, x^{a_j} =0,  \; \ i=1,\dots,d.
 \end{equation}
We also consider the following family of polynomial systems parametrized
 by $\gamma \in \R_{>0}^n$:
 \begin{equation}\label{systemgamma}
 \sum_{j=1}^n c_{ij}\, \gamma_{j} \, x^{a_j}=0 , \;  \ i=1,\dots,d.
 \end{equation}
Note that each polynomial system in the families~\eqref{systemwitht} and~\eqref{systemgamma} has again support $\mathcal{A}$.

\begin{thm}[Theorem 3.4 of \cite{bihan} and Theorem 2.11 of \cite{BDG1}]\label{th:BS}
Consider  $\Delta_1, \dots, \Delta_p$  $d$-simplices which occur in a regular subdivision  $\Gamma$ of a finite configuration
$\mathcal A \subset \Z^d$, and which are positively decorated 
by a matrix $C  \in \R^{d \times n}$. 
\begin{enumerate}
\item Let $h$ be any height function that defines $\Gamma$.
Then, there exists a positive real number $t_0$ such that for all $0<t<t_0$, the number of nondegenerate positive
solutions of (\ref{systemwitht}) is at least $p$.
\item Let $m_1\dots,m_\ell \in \R^n$ be vectors that define a presentation of the cone ${\mathcal C}_{\Delta_1,\dots, \Delta_p}$ :
 \begin{equation*} \label{eq:cone}
 {\mathcal C}_{\Delta_1,\dots, \Delta_p} \, = \, \{h \in \R^n \, : \, \langle m_j, h \rangle  > 0 , \; j=1,\ldots,\ell\}.
 \end{equation*}
Then, for any $\varepsilon \in (0,1)^\ell$ there exists $t_0(\varepsilon) >0$ such that for any $\gamma$ in the open set
\[ U \, = \, \cup_{\varepsilon \in (0,1)^\ell} \,  \{ \gamma \in \R_{>0}^n \, ; \, \gamma^{m_j} < t_0(\varepsilon)^{\varepsilon_j}, \, j=1 \dots,\ell\},\]
 system~\eqref{systemgamma}
 has at least $p$ nondegenerate positive solutions.
\end{enumerate}
\end{thm}

We remark that in the first item in Theorem~\ref{th:BS} (Theorem 3.4 in~\cite{bihan}), we describe a piece of curve in the space
of coefficients as we vary $t>0$ where the associated system has at least $p$ positive solutions, 
while in the second item in Theorem~\ref{th:BS} (Theorem 2.11 in~\cite{BDG1}), we describe a 
subset with nonempty interior in the space of coefficients where we can bound from below by $p$
the number of positive solutions of the associated system.

\section{Proofs of Theorems~\ref{th:allintermediatesEside}  and~\ref{th:J}}\label{sec:proofs}

We start this section with a lemma.

\begin{lemma}\label{lemma:Tisregular} Consider $\mathcal{A}=\{(1,0),(0,1),(1,1),(1,2),\dots,(1,n),(0,0)\}\subset \Z^2$. 
The triangulation $T$ of $\mathcal A$ with the following $2$-simplices:
\[\{ \{(1,j),(1,j+1),(0,0)\},j=0\dots n-1, \{(0,1),(1,n),(0,0)\} \}\] 
 is regular (see Figure~\ref{fig:TriangulationT}).
\begin{center}
\begin{figure}[h]
	\begin{tikzpicture}
	[scale=1.300000,
	back/.style={loosely dotted, thin},
	edge/.style={color=black, thick},
	facet/.style={fill=red!70!white,fill opacity=0.800000},
	vertex/.style={inner sep=1pt,circle,draw=black,fill=blue,thick,anchor=base
	}]
	%
	%
	\coordinate (0.00000, 0.00000) at (0.00000, 0.00000);
	\coordinate (0.00000, 1.00000) at (0.00000, 1.00000);
	\coordinate (1.00000, 0.00000) at (1.00000, 0.00000);
    \fill[facet, fill=blue!30!white,fill opacity=0.800000] (0,0) -- (1,0) -- (1,1)-- cycle {}; \fill[facet, fill=blue!30!white,fill opacity=0.800000] (0,0) -- (1,1) -- (1,2)-- cycle {}; \fill[facet, fill=blue!30!white,fill opacity=0.800000] (0,0) -- (1,3.1) -- (1,4.1)-- cycle {}; \fill[facet, fill=blue!30!white,fill opacity=0.800000] (0,0) -- (0,1) -- (1,4.1)-- cycle {};

%
	\draw[edge] (0,0) -- (0,1);
	\draw[edge] (0,0) -- (1,0);
	\draw[edge] (0,0) -- (1,1);
	\draw[edge] (0,0) -- (1,2);
	\draw[edge] (0,0) -- (1,3.1);
	\draw[edge] (0,0) -- (1,4.1);
	\draw[edge] (1,0) -- (1,2);
	\draw[edge,dotted] (1,2) -- (1,3.1);
	\draw[edge] (1,3.1) -- (1,4.1);
	\draw[edge] (0,1) -- (1,4.1);
	
	\node[vertex] at (0,1){};
	\node[vertex] at (1,1){};
	\node[vertex] at (1,0){};
	\node[vertex] at (0,0){};
	\node[vertex] at (1,2){};
	\node[vertex] at (1,3.1){};
	\node[vertex] at (1,4.1){};

	\end{tikzpicture}
	\caption{Triangulation $T$ of $\mathcal{A}$.} \label{fig:TriangulationT}
	\end{figure}
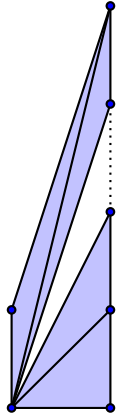
\end{center}
\end{lemma}
\begin{proof}

We can take $h\colon\mathcal{A}\to \R$, with $h(0,0)=0$, $h(0,1)=n$ and 
$h(1,j)=\frac{j(j-1)}{2}$, for $j=0,\dots,n-1$. It is easy to check $h$ defines a regular triangulation that is equal to $T$. 

\end{proof}

The idea in the proofs of Theorem~\ref{th:allintermediatesEside} and Theorem~\ref{th:J} is to detect positively
 decorated simplices in the regular triangulation $T$.

\begin{proof}[Proof of Theorem~\ref{th:allintermediatesEside}]
By Proposition~\ref{thm:upperboundGJ}, the number of positive solutions of the system $G_{I_n}$ is at most $2 [\frac{n}{2}]+1$. So, it is enough to prove that this number is also a lower bound.

The number of positive steady states of the system $G_{I_n}$ is the number of positive solutions of the system~\eqref{eq:systemG_J}. As we saw before, the support of this last system is 
{\small
\[\mathcal{A}=\{(1,0),(0,1),(1,1),(1,2),\dots,(1,n),(0,0)\}\subset \Z^2,\]}%
with coefficient matrix $C$~\eqref{eq:matrixC}. Note that if one multiplies a column of $C$ by a positive number, then a simplex is positively decorated by $C$ if and only if it is positively decorated by the new matrix. After multiplying the columns by convenient positive numbers, we obtain the following matrix from $C$:
\[Csimple=\begin{pmatrix}
1&0&M_0&\dots &M_{n-1} &-S_{tot}\\
0&1&1&\dots   &1   &-E_{tot}
\end{pmatrix},\]
where $M_i=\frac{k_{\rm{cat}_i}}{\nu_iF_{tot}}+1$, for each $i=0,\dots,n-1$. We will work with this new matrix $Csimple$.

We consider the regular triangulation $T$ in 
Lemma~\ref{lemma:Tisregular}. 
The simplex  $\{(1,0),(1,1),(0,0)\}$ of $T$ is positively decorated by $Csimple$ if and only if
$E_{tot}M_0-S_{tot}<0.$
The simplex $\{(1,j),(1,j+1),(0,0)\}$, for $j=1,\dots,n-1$, corresponds to the submatrix
\[Csimple_j=\begin{pmatrix}
M_{j-1}& M_{j}& -S_{tot}\\
1&1& -E_{tot}
\end{pmatrix},\]
and it is positively decorated by $Csimple$ if and only if $E_{tot}M_{j-1}-S_{tot}$ and $E_{tot}M_{j}-S_{tot}$ have opposite signs.
The last simplex $\{(0,1),(1,n),(0,0)\}$ is positively decorated by $Csimple$ if and only if $E_{tot}M_{n-1}-S_{tot}>0.$

Therefore we always have at least $n$ positively decorated simplices using all simplices of $T$  but the last one, just by imposing
{\small
\begin{equation}\label{ineq:proof}
(E_{tot}M_i-S_{tot})(-1)^i<0,\ \text{ for } i=0,\dots,n-1.
\end{equation}}%
We can include the last simplex if and only if $n$ is even (because otherwise the inequalities are not compatible), and in this case we have at least $n+1$ positively decorated simplices. We can obtain $2 [\frac{n}{2}]+1$ positively decorated simplices if the inequalities \eqref{ineq:proof} are satisfied. These inequalities are equivalent to the inequalities \eqref{ineq:allintermediates} in the statament.

Assume \eqref{ineq:allintermediates} holds. Given any height function $h$ inducing the triangulation $T$, by item (1) in Theorem~\ref{th:BS} there exists $t_0$ in $\R_{>0}$
such that for all $0<t<t_0$, the number of positive nondegenerate solutions of the deformed system as in~\eqref{systemwitht} with support $\mathcal{A}$ 
and coefficient matrix $C_t$, with $(C_t)_{ij}=t^{h(\alpha_j)}c_{ij}$ (with $\alpha_j\in\mathcal{A}$, $C=(c_{ij})$), is at least $2 [\frac{n}{2}]+1$. In particular if we choose $h$ as in the proof of Lemma~\ref{lemma:Tisregular}, there exists $t_0$ in $\R_{>0}$, 
such that for all $0<t<t_0$, the system
\begin{small}
\begin{eqnarray}
s_0 + \sum_{j=0}^{n-1} \left(\frac{T_{j}}{(F_{tot})^{j+1}} + \frac{K_j\, T_{j-1}}{(F_{tot})^j}\right)\,t^{\frac{j(j+1)}{2}} {s_0e^{j+1}} - S_{tot}=0,\\
t^n e + \sum_{j=0}^{n-1} \frac{K_j\, T_{j-1}}{(F_{tot})^j}\, t^{\frac{j(j+1)}{2}}{s_0e^{j+1}} - E_{tot}=0,\nonumber
\end{eqnarray}
\end{small}%
has at least $2 [\frac{n}{2}]+1$ positive solutions. If we change the variable $\bar{e}=t^n e$, 
we get the following system:
\begin{small}
\begin{eqnarray}\label{system:allintermediateswitht}
s_0 + \sum_{j=0}^{n-1} \left(\frac{T_{j}}{(F_{tot})^{j+1}} + \frac{K_j\, T_{j-1}}{(F_{tot})^j}\right)\,t^{(j+1)
(\frac{j}{2}-n)} {s_0\bar{e}^{j+1}}  - S_{tot}=0,\\
\bar{e} + \sum_{j=0}^{n-1} \frac{K_j\, T_{j-1}}{(F_{tot})^j}\, t^{(j+1)(\frac{j}{2}-n)}{s_0\bar{e}^{j+1}} - E_{tot}=0. \nonumber
\end{eqnarray}
\end{small}%

It is straightforward to check that if we scale the constants $K_j$ by
{\small
\begin{equation}\label{eq:rescaling1} t^{j-n}K_j, \quad j=0,\dots,n-1, 
\end{equation}}%
while keeping fixed the values of the constants $k_{\rm{cat}_j}$, $\nu_j$ for $j=0,\dots,n-1$ and the values of the linear conservation
constants $E_{tot}$, $F_{tot}$ and $S_{tot}$ (assuming that condition~\eqref{ineq:allintermediates} holds), the intersection of the 
steady state variety and the linear conservation equations of the corresponding network is described by system~\eqref{system:allintermediateswitht}. 
It is easy to check that in order to get the scaling in~\eqref{eq:rescaling1}, it is sufficient to rescale only the original constants $k_{\rm{on}_j}$ as follows: 
$t^{j-n}k_{\rm{on}_j}$, for $j=0,\dots,n-1$. Then, for these choices of constants, the system has at least $2 [\frac{n}{2}]+1$ positive steady states.

Now, we will show how to obtain the more general rescaling in the statement. 
The existence of the positive constants $B_1, \dots, B_n$ follows from the inequalities that define the
cone $\mathcal{C}_T$ of heights inducing the regular triangulation $T$ and item (2) in Theorem~\ref{th:BS}. For instance, 
we can check that $\mathcal{C}_T$ is defined by $n$ inequalities:
\begin{equation*}
\mathcal{C}_T=\{h=(h_1,\dots,h_{n+3})\in \R^{n+3}\, : \, \langle m_j, h \rangle > 0,\, j=1,\dots,n\},
\end{equation*} 
where $\langle ,\! \rangle$ denotes the canonical inner product of $\R^{n+3}$ and $m_1=e_1-2e_3+e_4$, $m_j=e_{j+1} - 2e_{j+2} + e_{j+3}$, for $j=2,\dots,n-1$ and $m_n=e_2+e_{n+1}-e_{n+2}-e_{n+3}$, where $e_i$ denotes the $i$-th canonical vector of $\R^{n+3}$. 
Fix $\varepsilon\in (0,1)^{n+3}$. As \eqref{ineq:allintermediates} holds, item (2) in Theorem~\ref{th:BS} says that there exist positive numbers
$B_j$  for $j=1,\dots,n$ (depending on $\varepsilon$), such that the system
\begin{small}
\begin{eqnarray}
\gamma_1 s_0 + \sum_{j=0}^{n-1} \left(\frac{T_{j}}{(F_{tot})^{j+1}} + \frac{K_j\, T_{j-1}}{(F_{tot})^j}\right)\,\gamma_{j+3} {s_0e^{j+1}} - \gamma_{n+3} S_{tot}=0,\\
\gamma_2 e + \sum_{j=0}^{n-1} \frac{K_j\, T_{j-1}}{(F_{tot})^j}\, \gamma_{j+3}{s_0e^{j+1}} - \gamma_{n+3}E_{tot}=0,\nonumber
\end{eqnarray}
\end{small}%
has at least $2 [\frac{n}{2}]+1$ nondegenerate positive solutions, for any vector 
$\gamma \in \R^{n+3}$ satisfying $\gamma^{m_j}<B_j$, for all $j=1,\dots,n$. In particular,
 this holds if we take $\gamma_1 = \gamma_2 = \gamma_{n+3}=1$ and
\begin{equation}\label{eq:gamma-allintermediates}
\gamma_3^{-2}\gamma_4<B_1,\quad \gamma_{j+1}\gamma_{j+2}^{-2}\gamma_{j+3}<B_{j}, 
\text{ for } j=2,\dots,n-1, \quad \gamma_{n+1}\gamma_{n+2}^{-1}<B_n.
\end{equation}
If we call $\lambda_0=\gamma_3$ and $\lambda_{j}=\frac{\gamma_{j+3}}{\gamma_{j+2}}$ for $j=1,\dots,n-1$, 
the inequalities in \eqref{eq:gamma-allintermediates} are equivalent to the conditions \eqref{ineq:gammaallintermediates2}. 
Then, if $\lambda_j$, $j=0,\dots,n-1$, satisfy these inequalities, the rescaling of the given  parameters $k_{\rm{on}_j}$ 
by $\lambda_j k_{\rm{on}_j}$ for $j=0,\dots,n-1$, gives rise to a system with exactly $2 [\frac {n} {2}]+1$
positive steady states.
\end{proof}

The proof of Theorem~\ref{th:J} is similar to the previous one.
\begin{proof}[Proof of Theorem~\ref{th:J}] 
Again, the number of positive steady states of our system is equal to the number of positive solutions of the system~\eqref{eq:systemG_J}. Recall that the support of the system is  
{\small \[\mathcal{A}=\{(1,0),(0,1),(1,1),(1,2),\dots,(1,n),(0,0)\}\subset \Z^2.\]}%
In this case, the coefficient matrix $C$~\eqref{eq:matrixC} is equal, after multiplying the columns by convenient positive numbers, to the matrix
{\small
\[Csimple=\begin{pmatrix}
1&0&M_0&\dots &M_{n-1} &-S_{tot}\\
0&1&D_0&\dots   &D_{n-1}   &-E_{tot}
\end{pmatrix},\]}%
where $M_i=\frac{k_{\rm{cat}_i}}{\nu_iF_{tot}}+1$ and $D_i=1$, for each $i\in J$, and $M_i=1$ and $D_i=0$, for each $i\notin J$.

We consider again the regular triangulation $T$ in Lemma~\ref{lemma:Tisregular}. 
Recall that $J\subset J_n,$ see (\ref{Jn}), and therefore each $j\in J$ has the same parity as $n,$ in particular $0\leq j\leq n-2.$ 
For each $j\in J$, consider the  simplices $\Delta_j=\{(1,j),(1,j+1),(0,0)\}$ and 
$\Delta_{j+1}=\{(1,j+1),(1,j+2),(0,0)\}$. Note that if $j\neq j'$ then $\{\Delta_j,\Delta_{j+1}\}$ and $\{\Delta_{j'},\Delta_{j'+1}\}$  are disjoint since $j,j'$ and $n$ have the same parity. 

The simplices are positively decorated by $Csimple$ (and then by $C$) if and only if the submatrices 
{\small
\[Csimple_j=\begin{pmatrix}
1& M_{j}& -S_{tot}\\
0&1& -E_{tot}
\end{pmatrix},\quad Csimple_{j+1}=\begin{pmatrix}
M_{j}& 1& -S_{tot}\\
1&0& -E_{tot}
\end{pmatrix},\]}%
are positively spanning, and this happens if and only if $E_{tot}M_j-S_{tot}<0$, where $M_j=\frac{k_{\rm{cat}_j}}{\nu_jF_{tot}}+1$, since $j\in J$. 
The simplex $\Delta_{n}=\{(0,1),(1,n),(0,0)\}$ 
is trivially positively decorated by $Csimple$. Then, by imposing the inequalities $E_{tot}M_j-S_{tot}<0$ for $j\in J$, which are equivalent to the ones in the statement~\eqref{ineq:thJ}, we can obtain $2|J|+1$ positively decorated simplices.

Assume \eqref{ineq:thJ} holds. Given any height function $h$ inducing the triangulation $T$, by item (1) in Theorem~\ref{th:BS} there exists $t_0$ in $\R_{>0}$, 
such that for all $0<t<t_0$, the number of positive nondegenerate solutions of the deformed system with support $\mathcal{A}$
 and coefficient matrix $C_t$, with $(C_t)_{ij}=t^{h(\alpha_j)}c_{ij}$ (with $\alpha_j\in\mathcal{A}$, $C=(c_{ij})$) is at least $2|J|+1$. In particular if we choose $h$ as in the proof of Lemma~\ref{lemma:Tisregular}, there exists $t_0$ in $\R_{>0}$,
  such that for all $0<t<t_0$, the system
\begin{footnotesize}
\begin{eqnarray}
s_0 + \sum_{j\in J} \left(\frac{T_{j}}{(F_{tot})^{j+1}} + \frac{K_j\, T_{j-1}}{(F_{tot})^j}\right)\,t^{\frac{j(j+1)}{2}} {s_0e^{j+1}} 
+ \sum_{j\notin J} \frac{T_{j}}{(F_{tot})^{j+1}} \, t^{\frac{j(j+1)}{2}} {s_0e^{j+1}} - S_{tot}=0,\\
t^n e + \sum_{j\in J} \frac{K_j\, T_{j-1}}{(F_{tot})^j}\, t^{\frac{j(j+1)}{2}}{s_0e^{j+1}} - E_{tot}=0,\nonumber
\end{eqnarray}
\end{footnotesize}%
has at least $2|J|+1$ positive solutions. If we change the variable $\bar{e}=t^n e$, we get the following system:
\begin{footnotesize}
\begin{eqnarray}\label{system:Jintermediateswitht}
s_0 + \sum_{j\in J} \left(\frac{T_{j}}{(F_{tot})^{j+1}} + \frac{K_j\, T_{j-1}}
{(F_{tot})^j}\right)\,t^{(j+1)(\frac{j}{2}-n)} {s_0\bar{e}^{j+1}} + \sum_{j\notin J} 
\frac{T_{j}}{(F_{tot})^{j+1}} \, t^{(j+1)(\frac{j}{2}-n)} {s_0\bar{e}^{j+1}} - S_{tot}=0,\nonumber\\
\bar{e} + \sum_{j\in J} \frac{K_j\, T_{j-1}}{(F_{tot})^j}\, t^{(j+1)(\frac{j}{2}-n)}{s_0\bar{e}^{j+1}} - E_{tot}=0.
\end{eqnarray}
\end{footnotesize}%
Similarly as we did in the previous proof, if we scale the original parameters $k_{\rm{on}_j}$, for $j\in J$, and $\tau_j \text{ if } j\notin J$ by 
{\small
\begin{equation}\label{eq:scalingJ}
t^{j-n}k_{\rm{on}_j} \text{ if } j\in J, \quad t^{j-n}\tau_j \text{ if } j\notin J,  
\end{equation}}%
respectively, and if
we keep fixed the values of the remaining rate constants and the 
values of the linear conservation constants $E_{tot}$, $F_{tot}$ and $S_{tot}$, the intersection of the steady state variety 
and the linear conservation relations is described by system \eqref{system:Jintermediateswitht}. 
Then, for these choices of constants the system $G_J$ has at least $2|J|+1$ positive steady states.
The general rescaling that appears in the statement can be obtained in a similar way as we did in the proof of Theorem~\ref{th:allintermediatesEside}.
\end{proof}

\section{Computer aided results} \label{sec:R}

In this section we explore a computational approach to the multistationarity problem, more precisely we find new 
regions of mulstistationarity. The idea is to find good regular triangulations of the point configuration corresponding 
to the exponents of the polynomial system given 
by the steady state equations and conservation laws,
and deduce from them some regions of multistationarity. We first give the idea and then apply 
it for the $n$-site phosphorylation system for $n=2,3,4,$ and $5$, where we have successfully
found several regions of multistationarity. This approach can be, in principle, applied to other systems 
if they satisfy certain hypotheses, see \cite[Theorem 5.4]{BDG1}, and are sufficiently small in order for the computations
 to be done in a reasonable amount of time. 

The strategy is the following. Given a polynomial system with support $\mathcal{A}$ and matrix of coefficients $C$, 
one first computes all possible regular triangulations of $\mathcal{A}$ with the aid of a computer. 
The number of such triangulations can be very large depending on the size of $\mathcal{A},$
 thus the next step is to discard in each such triangulation the simplices that obviously will not be positively decorated by $C.$ 
With the reduced number of triangulations one can now search through all of them for the ones 
giving the biggest number of simultaneously positively decorated simplices. Each set of $k$ simultaneously positively 
decorated simplices gives a candidate for a region of 
multistationarity with $k$ positive steady states. If one finds $m$ of such sets, then it is possible to have up to $m$ such regions.
Have in mind, however, that among these regions can be repetitions.

Next we apply this to the $n$-site phosphorylation system with all intermediates and explain more concretely this procedure in this case.

In Corollary \ref{cor:1/4} we obtained regions of multistationarity with $2[\frac n 2] +1$ positive steady states each using only $1/4$ of the intermediates, 
our objective now is to understand if it is possible to find more such regions with more intermediates. 
Consider the network $G$ of the $n$-site phosphorylation system with all possible intermediates:
 
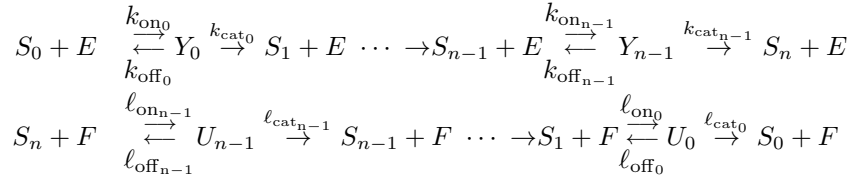
\begin{small}
 \[   
 \begin{array}{rl} 
 \nonumber 
 S_0+E &
 \arrowschem{k_{\rm{on}_0}}{k_{\rm{off}_0}} Y_0
 \stackrel{k_{\rm{cat}_0}}{\rightarrow} S_1+E\ \cdots\ {\rightarrow}S_{n-1}+E 
 \arrowschem{k_{\rm{on}_{n-1}}}{k_{\rm{off}_{n-1}}}Y_{n-1}
 \stackrel{k_{\rm{cat}_{n-1}}}{\rightarrow} S_n+ E \\ 
 S_n+F &
 \arrowschem{\ell_{\rm{on}_{n-1}}}{\ell_{\rm{off}_{n-1}}} U_{n-1}
 \stackrel{\ell_{\rm{cat}_{n-1}}}{\rightarrow} S_{n-1}+F \ \cdots\ {\rightarrow} S_{1}+F
 \arrowschem{\ell_{\rm{on}_0}}{\ell_{\rm{off}_0}} U_0
 \stackrel{\ell_{\rm{cat}_0}}{\rightarrow} S_0+ F
 \end{array}
 \]
\end{small}%
In Section 4 of \cite{BDG1}, the concentration at steady state of all species are given in terms of the species $s_0, e, f$:
\[ \begin{array}{lcrl}
 \label{parametrizationphospo}
 s_i &=&T_{i-1} \, \frac{s_0e^i}{f^i},& i=1,\dots,n, \nonumber \\
 y_i &=& K_i \, T_{i-1} \, \frac{s_0e^{i+1}}{f^i},& i=0,\dots,n-1, \\
 u_i &=& L_i \, T_i \, \frac{s_0e^{i+1}}{f^i}, &  i=0,\dots,n-1, \nonumber
 \end{array}\]
where $K_i=\frac{k_{\rm{on}_i}}{k_{\rm{off}_i}+k_{\rm{cat}_i}}, L_i=
\frac{\ell_{\rm{on}_i}}{\ell_{\rm{off}_i}+\ell_{\rm{cat}_i}}, T_i=\prod^i_{j=0}\frac{\tau_j}{\nu_j}$ 
 for each $i=0,\dots,n-1$ (recall that $K_i^{-1}$ and $L_i^{-1}$ are usually called \textit{Michaelis-Menten constants})
  and $T_{-1}=1$, where $\tau_i=k_{\rm{cat}_i}K_i$ and $\nu_i=\ell_{\rm{cat}_i}L_i$, for each $i=0,\dots, n-1$.
 
This looks very similar to \eqref{parametrizationGJ}, where $F_{tot}$ is replaced by $f$, which is now a variable, and
so we need to work in dimension $3$. The main difference is that the networks $G_J$ considered in the previous sections have intermediates only in the 
$E$ component and the network $G$ we consider in this section has all intermediates.

Note from \eqref{eq:consfosfo} that the support $\mathcal{A}$ of this system, which has $2n+4$ elements, ordering 
the variables as $s_0,e,f,$ is given by the columns of the following matrix
\[A=
\begin{small}
\left(\begin{array}{cccrrrcrccc}
1&0&0&   1& \dots &1&    1&1&\dots &1&0\\
0&1&0&   1& \dots &n&    1&2&\dots &n&0\\
0&0&1&  -1& \dots &-n&   0&-1&\dots &1-n&0
\end{array} \right)\end{small},\]%
and the corresponding matrix of coefficients for the system is
\[
C=
\begin{footnotesize}
\begin{pmatrix}
1&0&0&   T_0&\dots &T_{n-1}&
    K_0+L_0T_0&K_1T_0+L_1T_1&\dots 
    &K_{n-1}T_{n-2}+L_{n-1}T_{n-1}&-S_{tot}\\
0&1&0&   0&\dots &0& 
   K_0&K_1T_0&\dots &K_{n-1}T_{n-2}&-E_{tot}\\
0&0&1&   0&\dots &0&
   L_0T_0&L_1T_1&\dots &L_{n-1}T_{n-1}&-F_{tot}
\end{pmatrix}
\end{footnotesize}.
\]

Recall that if one multiplies a column of a matrix $C$ by a positive number, then a simplex is positively decorated 
by $C$ if and only if it is positively decorated by the modified matrix. So, in order to test whether a simplex with vertices in 
$\mathcal{A}$ is positively decorated by $C$ is enough to test if it is positively decorated by the following matrix
\[
Csimple=
\begin{small}\begin{pmatrix}
1&0&0&   1&\dots &1&
    1&1&\dots 
    &1&-S_{tot}\\
0&1&0&   0& \dots &0& 
   N_0&N_1&\dots &N_{n-1}&-E_{tot}\\
0&0&1&   0&\dots &0&
   1-N_0&1-N_1&\dots &1-N_{n-1}&-F_{tot}
\end{pmatrix},\end{small}
\]%
where $0<N_i=
\dfrac{K_iT_{i-1}}{K_iT_{i-1}+L_iT_{i}}=\left( 1+\dfrac{k_{\rm{cat}_i}}{l_{\rm{cat}_i}}\right)^{-1}<1$ for $i=0,1,\dots,n-1.$ 
Here the matrix $Csimple$ is obtained by dividing the fourth until the last column by its first entry.

Now we compute all possible regular triangulations of $\mathcal{A}$ and search through them looking for the ones 
with the maximal possible number of simplices simultaneously positively decorated by $Csimple.$ Since the number 
of such triangulations grows very fast with $n$ we approach it with the following strategy:

\begin{algorithm}\label{algo}
\begin{enumerate}
\item Compute $L_1:=\{\mbox{all possible triangulations of } \mathcal{A}\}.$\footnote{We are calling $L_1, \dots, L_7$ the sets defined in Algorithm~\ref{algo}. They are completely unrelated to the rational functions of the rate constants denoted with the same letters.}  
\item With $L_1$ compute $L_2$ by discarding
all simplices which do not have the last vertex $(0,0,0).$
In fact we only need these simplices since a simplex not containing the last vertex cannot be positively decorated, 
because the corresponding coefficients of $Csimple$ will be all positive. 
\item Compute $L_3$ from $L_2$ by removing all simplices with the corresponding $3\times 4$ submatrix of $Csimple$ 
having a zero $3\times 3$ minor. The reason for this is clear, such simplices will never be positively decorated by $Csimple.$ 
\item Compute $L_4$ from $L_3$ using the symmetry of $Csimple.$ More precisely, change any index $4,5,\dots,n+3$ on
 the simplex to $1$ because on $Csimple$ they yield the same column. Here we are using the easy-to-check fact that changing 
 the order of indexes does not change the conditions for a simplex to be positively decorated.
\item Compute $L_5$ from $L_4$ removing all $T\in L_4$ such that there is another $T'\in L_4$ with $T\subset T'.$ 
\item For each $T\in L_5,$ check for each set $S\subset T$ of simplices if there is viable $N_0,\dots,N_{n-1}$ such that all 
$\Delta\in S$ is positively decorated by $Csimple$ at the same time, call $L_6$ the list of such $S'$s.
\item If the maximum size of a element in $L_6$ is $k,$ set $L_7:=\{T\in L_6\, : \, \#T=k\}.$ 
This $k$ is the number of positive steady states and $m:=\#L_7$ is the number of candidates for regions of multistationarity.
\end{enumerate}
\end{algorithm}

Step (1) can be done with the package TOPCOM inside SAGE \cite{sage}, the other steps are quite simple to implement, 
for instance in MAPLE \cite{maple}. We show in the table below the number of elements in some of the lists and 
an approximation of the computation time for small values of $n.$
\begin{small}
\begin{table}[h!]
\begin{tabular}{l|c|c|c|c|c|c|c|c|c}
$n$ & $\# L_1$ & $\# L_2$ & $\# L_3$ & $\# L_4$ & $\# L_5$ &  $\# L_7$ &$k$ &
\makecell{regions of\\ multistationarity}&
\makecell{computation time} \\ \hline
$2$ &   44     &  25      &    15    &       7 &   6     &  1& 3& 1 &  negligible   \\ \hline
$3$ &   649     &     260   &  100      &       21 &      18  & 6 & 3& 6 & about 1 sec       \\ \hline
$4$ &    9094    &    2728    &   682     &       62 &   53     &  5 & 5 &4 & about 2 min      \\ \hline
$5$ &  122835      & 28044       &   4560     &     177   &      149  &   23 &   5& 15  & about 3 hours  
\end{tabular}
\end{table}
\end{small}

The most computationally expensive part is to compute all regular triangulations, taking at least 90\% of the time. 
These computations were done in a Linux virtual machine with 4MB of RAM and with 4 cores of 3.2GHz of processing. 
With a faster computer or more time one probably can do $n=6$ or even $n=7$ but probably not much more than this. 
For $n=5$ just the file for the raw list $L_1$ of regular triangulations already has $10$Mb.

An alternative path to Steps (6) and (7) is to set a number $k$ and look for sets $T\in L_5$ and $S\subset T$ with $\# S\geq k$ such that there is viable $N_0,\dots, N_{n-1}$ such that all $\Delta\in S$ are positively decorated by $Csimple$ at the same time. We actually used this with $k=2[\frac{n}{2}]+1.$ This other route depends upon a good guess one may previously have at how many positive steady states to expect.

After Step (7) one has to determine if there are 
any repetitions among the candidates 
for regions of multistationarity in $L_7$ and also if there are
any superfluous candidates of regions, that is conditions $C_1$ and $C_2$ such that $C_1$ implies $C_2.$ In our case we did it by hand since the $\# L_7$ was quite small.

Once Step (7) is done, one has a list of inequalities for each element $S$ of $L_7.$ These come from the conditions imposing that 
the simplices in $S$ are positively decorated by $Csimple.$
We are going to use these conditions to describe the regions of mulstistationarity. 
Because of the uniformity of $Csimple$ the only kind of conditions that appear are
\begin{small}
\begin{enumerate}[label = {$\rm{(\Roman*)}_i$}]
\item[$\rm{(I)}_{i,j}$] $N_i-N_j>0$
\setcounter{enumi}{1}
\item $S_{tot}N_i-E_{tot}>0 $
\item $E_{tot}N_i+F_{tot}N_i-E_{tot}>0 $
\item $-S_{tot}N_i-F_{tot}+S_{tot}>0 $
\item[$\rm{(V)}$]  $S_{tot}>E_{tot}+F_{tot}$,
\end{enumerate}
\end{small}%
or the opposite inequalities, and these translate from the $N_i$ to the $k_{\rm{cat}_i},\ell_{\rm{cat}_i}$ as follows
{\small\begin{enumerate}[label = {$\rm{(\Roman*)}_i'$}]
\item[$\rm{(I)}_{i,j}'$] $\dfrac{k_{\rm{cat}_j}}{\ell_{\rm{cat}_j}}>\dfrac{k_{\rm{cat}_i}}{\ell_{\rm{cat}_i}}$
\setcounter{enumi}{1}
\item $\dfrac{S_{tot}-E_{tot}}{E_{tot}}>\dfrac{k_{\rm{cat}_i}}{\ell_{\rm{cat}_i}}$
\item $\dfrac{F_{tot}}{E_{tot}}>\dfrac{k_{\rm{cat}_i}}{\ell_{\rm{cat}_i}} $
\item $\dfrac{F_{tot}}{S_{tot}-F_{tot}}<\dfrac{k_{\rm{cat}_i}}{\ell_{\rm{cat}_i}} .$
\end{enumerate}}%
Note that 
\begin{itemize}
\item conditions $\rm{(III)}_i$ and $\rm{(V)}$ together imply $\rm{(II)}_i;$
\item the opposite of condition $\rm{(II)}_i$ together with condition $\rm{(V)}$ imply the opposite of $\rm{(III)}_i$;
\item the opposite of condition $\rm{(III)}_i$ together with condition $\rm{(V)}$  imply $\rm{(IV)}_i;$ 
\item the opposite of condition $\rm{(IV)}_i$ together with condition $\rm{(V)}$  imply  $\rm{(III)}_i;$
\item condition $\rm{(III)}_i$ and the opposite of $\rm{(III)}_j$ together imply  $\rm{(I)}_{i,j}.$
\end{itemize}

Using these properties it is easy to describe in a nice manner the regions of multistationarity and discard the repeated and superfluous ones. 
We sum up our findings on the following results which are proved in the same fashion as Theorems~\ref{th:allintermediatesEside} and~\ref{th:J}, once you have 
the regular triangulation obtained with the computer script. In the following propositions we describe the regions of multistationarity for $n=2,3,4$ and $5.$

\begin{prop}\label{prop:maple2}
Let $n=2$. Assume that $S_{tot}>E_{tot}+F_{tot}.$ Then there is a choice of reaction rate constants for 
which the distributive sequential $2$-site phosphorylation system admits $3$ positive steady states. 
More explicitly, given rate constants and total concentrations such that 
{\small
\[\dfrac{k_{\rm{cat}_0}}{\ell_{\rm{cat}_0}}<\dfrac{F_{tot}}{E_{tot}}<
\dfrac{k_{\rm{cat}_1}}{\ell_{\rm{cat}_1}},\]}%
after rescaling of the $k_{\rm on}$'s and $\ell_{\rm on}$'s the distributive sequential $2$-site phosphorylation system has $3$ positive steady states.
\end{prop}

\begin{prop}\label{prop:maple3}
Let $n=3$. Assume that $S_{tot}>E_{tot}+F_{tot}.$
Then, there is a choice of rate constants for which the distributive sequential $3$-site phosphorylation system 
admits at least $3$ positive steady states. More explicitly, if the rate constants and total concentrations are in one of the regions described below
\begin{small}
\begin{enumerate}[label = {$\rm{(R_{3.\arabic*})}$}]
\item $\dfrac{k_{\rm{cat}_0}}{\ell_{\rm{cat}_0}}<\dfrac{F_{tot}}{E_{tot}}<
\dfrac{k_{\rm{cat}_1}}{\ell_{\rm{cat}_1}},$
\item $\dfrac{k_{\rm{cat}_0}}{\ell_{\rm{cat}_0}}<\dfrac{F_{tot}}{E_{tot}}<
\dfrac{k_{\rm{cat}_2}}{\ell_{\rm{cat}_2}},$
\item $\dfrac{k_{\rm{cat}_1}}{\ell_{\rm{cat}_1}}<\dfrac{F_{tot}}{E_{tot}}<
\dfrac{k_{\rm{cat}_2}}{\ell_{\rm{cat}_2}},$
\end{enumerate}
\end{small}%
then after rescaling of the $k_{\rm on}$'s and $\ell_{\rm on}$'s the distributive sequential $3$-site phosphorylation system has at least $3$ positive steady states.
\end{prop}

\begin{prop}\label{prop:maple3b}
Let $n=3$. If the rate constants and total concentrations are in one of the regions described below
\begin{small}
\begin{enumerate}[label = {$\rm{(R_{3.\arabic*})}$}]
\setcounter{enumi}{3}
\item  $\max\left\{
\dfrac{F_{tot}}{E_{tot}}, \dfrac{F_{tot}}{S_{tot}-F_{tot}}\right\}<\min\left\{\dfrac{k_{\rm{cat}_0}}{\ell_{\rm{cat}_0}},
\dfrac{k_{\rm{cat}_2}}{\ell_{\rm{cat}_2}}\right\}, \ S_{tot}>F_{tot}, $
\item  $\max\left\{
\dfrac{F_{tot}}{E_{tot}},\dfrac{F_{tot}}{S_{tot}-F_{tot}}\right\}<\min\left\{\dfrac{k_{\rm{cat}_1}}{\ell_{\rm{cat}_1}},
\dfrac{k_{\rm{cat}_2}}{\ell_{\rm{cat}_2}}\right\}, \  S_{tot}>F_{tot},  $
\item  $\min\left\{\dfrac{F_{tot}}{E_{tot}},\dfrac{S_{tot}-E_{tot}}{E_{tot}}\right\}>\max\left\{\dfrac{k_{\rm{cat}_1}}{\ell_{\rm{cat}_1}},
\dfrac{k_{\rm{cat}_2}}{\ell_{\rm{cat}_2}}\right\}, \  S_{tot}>E_{tot}, $
\end{enumerate}
\end{small}%
then after rescaling of the $k_{\rm on}$'s and $\ell_{\rm on}$'s the distributive sequential $3$-site phosphorylation system has at least $3$ positive steady states.
\end{prop}

\begin{prop}\label{prop:maple4}
Let $n=4$. Assume that $S_{tot}>E_{tot}+F_{tot}.$
Then, there is a choice of rate constants for which the distributive sequential $4$-site phosphorylation system has at least $5$ steady states. 
More explicitly, if the rate constants and total concentrations are in one of the regions described below
\begin{small}
\begin{enumerate}[label = {$\rm{(R_{4.\arabic*})}$}]
\item $
\dfrac{k_{\rm{cat}_2}}{\ell_{\rm{cat}_2}}<\dfrac{F_{tot}}{E_{tot}}<\min\left\{
\dfrac{k_{\rm{cat}_1}}{\ell_{\rm{cat}_1}},
\dfrac{k_{\rm{cat}_3}}{\ell_{\rm{cat}_3}}\right\},$
\item $\dfrac{k_{\rm{cat}_0}}{\ell_{\rm{cat}_0}}<\dfrac{F_{tot}}{E_{tot}}<\min\left
\{\dfrac{k_{\rm{cat}_1}}{\ell_{\rm{cat}_1}},
\dfrac{k_{\rm{cat}_3}}{\ell_{\rm{cat}_3}}\right\},$
\item $\max\left\{\dfrac{k_{\rm{cat}_0}}{\ell_{\rm{cat}_0}},
\dfrac{k_{\rm{cat}_2}}{\ell_{\rm{cat}_2}}\right\}<\dfrac{F_{tot}}{E_{tot}}<
\dfrac{k_{\rm{cat}_3}}{\ell_{\rm{cat}_3}},$
\item $\max\left\{\dfrac{k_{\rm{cat}_0}}{\ell_{\rm{cat}_0}},
\dfrac{k_{\rm{cat}_2}}{\ell_{\rm{cat}_2}}\right\}<\dfrac{F_{tot}}{E_{tot}}<
\dfrac{k_{\rm{cat}_1}}{\ell_{\rm{cat}_1}},$
\end{enumerate}
\end{small}%
then after rescaling of the $k_{\rm on}$'s and $\ell_{\rm on}$'s the distributive sequential $4$-site phosphorylation system has at least $5$ steady states.
\end{prop}

\begin{prop}\label{prop:maple5}
Let $n=5$. Assume that $S_{tot}>E_{tot}+F_{tot}.$
Then, there is a choice of rate constants for which the distributive sequential $5$-site phosphorylation system has at least $5$ steady states. 
More explicitly, if the rate constants and total concentrations are in one of the 13 regions described below
{\small
\begin{equation*}
(\emph{R}_{5.(I,J)}) \quad \quad
\max_{i\in I}\left\{\dfrac{k_{\rm{cat}_i}}{\ell_{\rm{cat}_i}}\right\}<
\dfrac{F_{tot}}{E_{tot}}<\min_{j\in J}\left\{\dfrac{k_{\rm{cat}_j}}{\ell_{\rm{cat}_j}}\right\},
\end{equation*}}%
with $(I,J)$ in the  following list {\rm (}where we write e.g. $14$ instead of $\{1,4\}${\rm )}:
\begin{small}
\begin{equation*}
(0,14),(0,24),
(1,24),(2,13),
(2,14),(3,14),(3,024),
(02,3),(02,4),
(03,1),(03,2),
(13,2),(13,4),
\end{equation*}
\end{small}%
%
then after rescaling of the $k_{\rm on}$'s and $\ell_{\rm on}$'s the distributive sequential $5$-site phosphorylation system has at least $5$ steady states.
\end{prop}

\begin{prop}\label{prop:maple5b}
Let $n=5$. If the rate constants and total concentrations are in one of the regions described below
\begin{small}
\begin{enumerate}[label = {$\rm{(R_{5.\arabic*})}$}]
\item  $\max\left\{
\dfrac{F_{tot}}{E_{tot}},
\dfrac{F_{tot}}{S_{tot}-F_{tot}}
\right\} <\min\left\{
\dfrac{k_{\rm{cat}_1}}{\ell_{\rm{cat}_1}},
\dfrac{k_{\rm{cat}_2}}{\ell_{\rm{cat}_2}},
\dfrac{k_{\rm{cat}_4}}{\ell_{\rm{cat}_4}}\right\}, \ S_{tot}>F_{tot},  $
\item  $\min\left\{
\dfrac{F_{tot}}{E_{tot}},
\dfrac{S_{tot}-E_{tot}}{E_{tot}}
\right\} >\max\left\{
\dfrac{k_{\rm{cat}_0}}{\ell_{\rm{cat}_0}},
\dfrac{k_{\rm{cat}_2}}{\ell_{\rm{cat}_2}},
\dfrac{k_{\rm{cat}_3}}{\ell_{\rm{cat}_3}}\right\}, \ S_{tot}>E_{tot}, $
\end{enumerate}
\end{small}%
then after rescaling of the $k_{\rm on}$'s and $\ell_{\rm on}$'s the distributive sequential $5$-site phosphorylation system has at least $5$ positive steady states.
\end{prop}

Note that the conditions in this section describe {\em different} regions from the ones described by the inequalities in 
Theorem~\ref{th:allintermediatesEside} and Theorem~\ref{th:J}. For instance, in 
Propositions~\ref{prop:maple2},~\ref{prop:maple3},~\ref{prop:maple4},~\ref{prop:maple5} 
the inequalities between the reaction rate constants and total conservations constants do not involve the value of $S_{tot}$. In Propositions~\ref{prop:maple3b} and~\ref{prop:maple5b}, the conditions onse
 the linear conservation constants are also different (e.g. on $\frac{F_{tot}}{E_{tot}}$ and $\frac{S_{tot}}{E_{tot}}-1$ instead of the product
 $F_{tot} (\frac{S_{tot}}{E_{tot}}-1)$). The inequalities  in Theorem~\ref{th:allintermediatesEside} and Theorem~\ref{th:J} hold for reactions rate constants
 of a reduced system $G_J$, but if we use Theorem~\ref{th:lifting} to extrapolate these conditions to the full $n$-site phosphorylation network, 
 the regions are different as well.

\section*{Discussion}

We developed in this paper both the theoretical setting based on~\cite{BDG1,bihan}  
and the algorithmic approach that follows from it,
to describe multistationarity regions in the space of all parameters for subnetworks of the $n$-site sequential phosphorylation cycle,
 where there are up to $2 [\frac n 2]+1$ positive steady states with fixed linear conservation constants. 
  We chose to assume that the subnetworks we consider have intermediate species  only in the $E$ component, but of course there
is a symmetry in the network interchanging $E$ with $F$, each $S_i$ with $S_{n-i}$, the corresponding intermediates and
rate constants, and completely similar results hold if we assume that there are only intermediates in the $F$ component.
 These regions can be lifted to regions of multistationarity with at least these number of stoichiometrically compatible positive steady states
 of the full $n$-site sequential phosphorylation cycle. The main feature of our {\em polyhedral} approach is that it gives a systematic
 method applicable to other networks (for instance, to enzymatic cascades~\cite{cascades}), and it provides {\em open 
 conditions} and not only {\em choices} of parameter values where high multistationarity occurs.  Moreover, we show that
 it can be algorithmically implemented.
 
 For instance,
 Theorem 1 in \cite{sontag} states that for any positive values of $S_{tot}$, $E_{tot}$ and $F_{tot}$, 
there exists $\varepsilon_0>0$  such that if $\frac{E_{tot}}{S_{tot}}<\varepsilon_0$, then there exists a choice of rate constants such 
that the system admits $2 [\frac n 2]+1$ positive steady states.
 Also, in the recent paper~\cite{FRW19}, the authors prove the
 existence of parameters for which there are $2 [\frac n 2]+1$ positive steady states and nearly half of them are
 stable and the other half unstable, using a different degeneration of the
 full network. However, the maximum expected number of stoichiometrically compatible steady states
 for the $n$-site system is equal to $2n-1$ (this is an upper bound by~\cite{sontag}), 
 which has been verified for $n=3$ and $4$ in~\cite{2n-1}. Probably, it is not
 possible to find a region in parameter space with this number of positive steady states using degeneration techniques.

In \cite{CM14} and \cite{CFMW17} it is proved that if $k_{\rm{cat}_0}/\ell_{\rm{cat}_0}<
k_{\rm{cat}_1}/\ell_{\rm{cat}_1}$
then there exist 
constants $E_{tot},F_{tot},S_{tot}$
such that the distributive sequential $2$-site phosphorylation presents multistationarity.
In \cite{Feliu18} the author gives new regions of multistationarity for $n=2,$ more precisely 
it is shown in Theorem 1 that given parameters $K_1,L_0,L_1,k_{\rm{cat}_0},k_{\rm{cat}_1},\ell_{\rm{cat}_0},\ell_{\rm{cat}_1},$ for any
$K_0$ small enough there exist 
constants $E_{tot},F_{tot},S_{tot}$
such that the distributive sequential $2$-site phosphorylation presents multistationarity. How small $K_0$ has to be taken is explicitly given in 
terms of the other parameters by a rather involved equation
(a similar result is obtained interchanging $K_0$ and $L_1$).   
Our Proposition~\ref{prop:maple2} gives a similar result, 
but with the advantage that we only need to 
rescale $k_{\rm on_0},k_{\rm on_1},\ell_{\rm on_0},\ell_{\rm on_1}.$

Proposition \ref{prop:maple2} is in agreement with \cite[Corollary 4.11]{CAT18} which establishes that in order for the  
distributive sequential 
$2$-site phosphorylation to present multistationarity the total
concentration of the substrate needs to be larger than either the concentration of the kinase or 
the phosphatase ($S_{tot}>F_{tot}$ or $S_{tot}>E_{tot}$). In the regions of multistationarity we found for $n=3,4$ and $5$ this is the case as well.
Propositions \ref{prop:maple2}, \ref{prop:maple3}, \ref{prop:maple3b}, \ref{prop:maple4}, \ref{prop:maple5} and~\ref{prop:maple5b} are of the same flavor as 
Theorems \ref{th:allintermediatesEside} and \ref{th:J}, in the sense that all of them give sufficient conditions on the rate 
constants and linear conservation constants such that after rescaling of the $k_{\rm on}$'s and $\ell_{\rm on}$'s, the distributive sequential $n$-site 
phosphorylation system is multistationary. The computational approach described in Section~\ref{sec:R} 
can be used for other networks of moderate size.

In conclusion, to the best of our knowledge other results 
(for instance \cite{CFMW17},\cite{CM14} and \cite{Feliu18}), only give regions of multistationarity as precise as ours for the distributive  
sequential $2$-site 
phosphorylation system. The present work has several advantages: it gives new regions of multistationarity for the distributive sequential $n$-site phosphorylation system for any $n;$
for small $n$ it gives several distinct regions of multistationarity; it gives lower bounds in the number of positive steady states in each of such regions; and moreover, our approach is easily applicable to other networks, both analytically and computationally.

\section*{Acknowledgements}
We are very thankful to Eugenia Ellis and Andrea Solotar for organizing the excellent project ``Matem\'aticas en el Cono Sur'', which
lead to this work. We also thank Eugenia Ellis for the warm hospitality in Montevideo, Uruguay, on December 3-7, 2018, where we devised
 our main results. AD, MG and MPM were partially supported by UBACYT 20020170100048BA, 
CONICET PIP 11220150100473 and 11220150100483, and ANPCyT PICT 2016-0398, Argentina.

\end{document}